\documentclass[3p,10pt,a4paper,twoside,fleqn,sort&compress]{elsarticle}
\usepackage{amssymb,amsmath,latexsym}
\usepackage{pst-all}
\usepackage[varg]{pxfonts}
\usepackage{mathrsfs}
\newtheorem{theorem}{Theorem}[section]
\newtheorem{lemma}[theorem]{Lemma}
\newtheorem{corollary}[theorem]{Corollary}
\newtheorem{remark}[theorem]{Remark}

\newtheorem{example}[theorem]{Example}

\def\endproof{\qed\endtrivlist}
\expandafter\let\csname endproof*\endcsname=\endproof

\def\qedsymbol{\ifmmode\bgroup\else$\bgroup\aftergroup$\fi
  \vcenter{\hrule\hbox{\vrule height.6em\kern.6em\vrule}\hrule}\egroup}
\def\qed{\ifmmode\else\unskip\nobreak\fi\quad\qedsymbol}

\renewcommand{\iff}{\Leftrightarrow}
\renewcommand{\implies}{\Rightarrow}

\renewcommand{\le}{\leqslant}

\newcommand{\dom}{\qopname\relax{no}{Dom}}
\newcommand{\im}{\qopname\relax{no}{Im}}

\begin{document}

\journal{\qquad}

\title{\Large\bf Nondeterministic automata: equivalence, bisimulations,\\ and uniform relations\tnoteref{t1}}
\tnotetext[t1]{Research supported by Ministry  of Science and Technological Development, Republic of Serbia, Grant No. 174013}
\author{Miroslav \'Ciri\'c\corref{cor}}
\ead{mciric@pmf.ni.ac.rs}

\author{Jelena Ignjatovi\'c}
\ead{jekaignjatovic73@gmail.com}

\author{Milan Ba\v si\'c}
\ead{basic\symbol{95}milan@yahoo.com}

\author{Ivana Jan\v ci\'c}
\ead{ivanajancic84@gmail.com}

\cortext[cor]{Corresponding author. Tel.: +38118224492; fax: +38118533014.}
\address{University of Ni\v s, Faculty of Sciences and Mathematics, Vi\v segradska 33, 18000 Ni\v s, Serbia}

\begin{abstract}\small
In this paper we study the equivalence of nondeterministic automata pairing the concept of a bisimulation with~the~recent\-ly introduced concept of a uniform relation.~In this symbiosis,
uniform relations serve as equivalence relations which relate states of two possibly different nondeterministic automata, and bisimulations ensure compatibility with the~tran\-sitions, initial and terminal states of these automata.~We define six types of bisimulations, but due to the duality we discuss three of them: forward, backward-forward, and weak forward bisimulations.~For each od these three types of bisimulations we provide a procedure which decides whether there is a bisimulation of this type between two automata, and when it exists,
the same procedure computes the greatest one.~We also show that there is a uniform~forward~bisimulation between two automata if and only if the factor automata with respect to the greatest forward bisimulation equivalences on these automata are isomorphic.~We prove a similar theorem for weak forward bisimulations, using the concept of a weak forward isomorphism instead of an isomorphism.~We also give examples that explain the~relationships between the considered types of bisimulations.

\end{abstract}


\maketitle

\section{Introduction}

One of the most important problems of automata theory is to determine whether two given automata~are equivalent, what usually means to determine whether their behaviour is identical.~In the context of deterministic or nondeterministic automata the behaviour of an~automaton is understood to be the language that is recognized by it, and two automata are considered~{\it equivalent\/}, or more precisely {\it language-equivalent\/},~if they recognize the same language.~For deterministic finite automata the equivalence problem is~solvable~in polynomial time, but for nondeterministic finite automata it is computationally hard (PSPACE-complete \cite{GJ.79,Skiena.08,Yu.97}).~Another important issue is to express the language-equi\-valence of two automata as a relation between their states, if such relationship exists, or find some kind of relations between states which would imply the language-equi\-valence.~The language-equi\-valence of two deterministic automata can be expres\-sed in terms of relationships between their states, but in the case of nondeterministic automata the problem is more complicated.

A widely-used notion of ``equivalence'' between states of automata is that of {\it bisimulation\/}.~Bisimulations have been introduced by Milner \cite{Milner.80} and Park \cite{Park.81} in computer science, where they~have been used to model equivalence between various~systems, as well as to reduce the number of states of these systems.~Roughly at the same time they have been also discovered in some areas of mathematics, e.g., in modal logic and set theory.~They are employed today in a many areas of computer science,~such~as functional languages, object-oriented languages, types, data types, domains, databases, compiler optimi\-zations, program analysis, verification tools, etc.~For more information about bisimulations we refer to \cite{AILS.07,CL.08,DPP.04,GPP.03,LV.95,Milner.89,Milner.99,RM-C.00,Sangiorgi.09}.

The most common structures on which bisimulations have been studied are labelled transition systems, i.e., labelled directed graphs, which are essentially nondeterministic automata without fixed initial and~terminal states.~A definition of bisimulations for nondeterministic automata that takes into account initial~and terminal states was given by Kozen in \cite{Kozen.97}.~In numerous papers dealing with bisimulations mostly one type of bisimulations has been studied, called just bisimulations, like in the Kozen's book \cite{Kozen.97}, or strong bisimulations, like in \cite{Milner.89,Milner.99,RM-C.00}.~In this paper we differentiate two types of simulations, forward and backward simulations.~Considering that there are four cases when a relation $R$ and its inverse $R^{-1}$ are forward or backward
simulations, we distinguish four types of bisimulations.~We define two homotypic~bisimulations, forward and backward bisimulations, where both $R$ and $R^{-1}$ are forward or backward simulations, and two heterotypic bisimulations, backward-forward and forward-backward bisimulations,  where $R$ is a backward and $R^{-1}$ a forward simulation or vice versa.~Distinction between forward and backward simulations, and forward and backward bisimulations, has been also made, for instance, in \cite{Buch.08,HKPV.98,LV.95} (for various kinds of automata), but~less or more these~concepts differ from the concepts having the same name which are considered here.~More~similar to our concepts of forward and backward simulations and bisimulations are those studied in \cite{Brihaye.07}, and in \cite{HMM.07,HMM.09} (for tree automata).

It is worth noting that forward and backward bisimulations, and backward-forward and forward-back\-ward bisimulations,~are dual concepts,
i.e.,
backward and forward-back\-ward bisimulations
on a nondeterministic automaton are forward and backward-forward bisimulations on its reverse automaton.~This~means that for any universally valid statement on forward or backward-forward bisimulations there is the corresponding universally valid statement on backward~and forward-backward bisimulations.~For that reason, our article deals only with forward and backward-forward bisimulations.~In  general, none of forward~and backward bisimulations or backward-forward and forward-back\-ward bisimulations can be considered in practical applications better than the other.~For example, under the names right and left invariant equivalences, forward and backward bisimulation equivalences
have been used by Ilie, Yu and others \cite{IY.02,IY.03,INY.04,ISY.05}  in reduction of the number of states of nondeterministic automata.~It was shown that there are cases where one of them better reduces the number of states,
but there are also other cases where the another one gives a better reduction.~There are also cases where each of them individually causes a polynomial reduction of the number of states, but alternately using both types of equivalences the number of states can be reduced exponentially (cf.~\cite[Section 11]{IY.03}).~It is also worth of mention that backward bisimulation equivalences were successfully applied in \cite{SCI.11} in the conflict analysis of discrete event systems, while it was shown that forward bisimulation equivalences can not be used for this purpose.

As we already said, the main role of bisimulations is to model equivalence between the states of the~same or different automata.~However, bisimulations provide compatibility with the transitions, initial and terminal states of automata, but in general they do not behave like equivalences.~A kind of relations~which can be conceived as equivalences which relate elements of two possibly different sets appeared~recently~in~\cite{CIB.09}~in the fuzzy framework.~Here we consider the crisp version of these relations, the so-called {\it uniform~relations\/}.
The main aim of the paper is to show that the conjunction of two concepts, uniform relations and~bisimulations, provides~a very powerful tool in the study of equivalence between nondeterministic automata, where uniform relations serve as equivalence relations
which relate states of two nondeterministic automata, and bisimulations ensure~com\-patibility with the transitions,
initial and terminal states of these automata.~Our second goal is to employ the calculus of relations as a tool that will show oneself as very effective in the study of bisimulations.~And third, we introduce and study a more general type of bisimulations, the so-called {\it weak bisimulations\/}.~We show that equivalence of automata determined by weak bisimulations is closer to the language equivalence than equivalence determined by bisimulations, and we also show that they produce smaller automata than bisimulations when they are used in the the reduction of the number of states.

Our main results are the following.~The main concepts and results from \cite{CIB.09} concerning uniform fuzzy relations are translated to the case of ordinary relations, and besides, the proofs and some statements are simplified (cf.~Theorems \ref{th:pur}, \ref{th:ur} and \ref{th:ur2}).~We also define the concept of the factor automaton with~respect to an arbitrary equivalence, and prove two theorems that can be conceived as a version, for nondetermi\-nistic automata, of two well-known theorems of universal algebra: Second Isomorphism Theorem and~Corre\-spondence Theorem (cf.~Theorems \ref{th:F:E} and \ref{th:F:E-isom}).~Then we study the general properties of forward and~back\-ward-forward bisimulations.~In cases where there is at least one forward or backward-forward bisimulation, we prove the existence of the greatest one, and we also show that the greatest forward bisimulation is a partial uniform relation (cf.~Theorems \ref{th:gfb} and \ref{th:gbfb}).~An algorithm that decides whether there is a forward bisimulation~between nondeterministic automata was provided by Kozen in \cite{Kozen.97}.~When there is a forward bisimulation, this algorithm also computes the greatest one.~Here we give another version of this~algorithm, and we also provide an analogous algorithm for backward-forward bisimulations (Theorems \ref{th:gfb.alg} and \ref{th:gbfb.alg}).

Given two automata $\cal A$ and~$\cal B$ and a uniform relation $\varphi \subseteq A\times B $ between their sets of states, we show that $\varphi$ is a forward bisimulation if and only if both its kernel $E_A^\varphi $ and co-kernel $E_B^\varphi $ are forward bisimulation~equi\-valences on $\cal A$ and $\cal B$, and the function $\widetilde \varphi $ induced in a natural way by $\varphi $ is an isomorphism between factor automata ${\cal A}/E_A^\varphi $ and ${\cal B}/E_B^\varphi $ (Theorem \ref{th:ufb}).~Also,~given two forward bisimulation equivalences $E$ on $\cal A$ and $F$ on $\cal B$, we show that there is a uniform forward~bisim\-ulation between $\cal A$ and $\cal B$ whose kernel and co-kernel are $E$ and $F$ if and only if the factor automata ${\cal A}/E$ and ${\cal B}/F$ are isomorphic (Theorem \ref{th:ufb.ex}).~Two automata $\cal A$ and $\cal B$ are defined to be FB-equivalent if there is a complete and surjective forward bisimulation between $\cal A$ and $\cal B$, which is equivalent to the existence of a uniform forward bisimulation between $\cal A$ and $\cal B$.~We prove that $\cal A$ and $\cal B$ are FB-equivalent if and only if the factor automata with respect to the greatest forward bisimulation equivalences on $\cal A$ and $\cal B$ are isomorphic (cf.~Theorem \ref{th:UFBeq}).~As a consequence we obtain that the factor automaton with respect to the greatest forward bisimulation equivalence on an automaton $\cal A$ is the unique (up to an isomorphism) minimal automaton in the class of all automata which are FB-equivalent to $\cal A$.~Let us note that similar results were proved in \cite{Kozen.97}, under the assumption that the automaton $\cal A$ is accessible, and in \cite{CCK.00}.

Theorems similar to Theorems \ref{th:ufb} and \ref{th:ufb.ex} are proved for backward-forward bisimulations (Theorems \ref{th:ubfb} and \ref{th:ubfb.ex}).~The only difference is that the kernel of a backward-forward bisimulation is a forward bisimulation equivalence, and the co-kernel is a backward bisimulation equivalence.~This difference is the reason why we can not use backward-forward bisimulations to define an equivalence relation between automata, but nevertheless, the existence of a backward-forward bisimulation between two automata implies the language equivalence between them.~As a tool for providing structural characterization of equivalence, backward-for\-ward bisimulations were used in \cite{BE.93}, and in \cite{BLS.05,BLS.06,BP.03,Buch.08,EK.01,EM.10,LS.05,Sak.09} within~the context of weighted~automata (under~different names).~We also prove that a function between the sets of states of two automata is a~forward bisimulation if and only if it is a backward-forward bisimulation (Theorem \ref{th:func}).

Then we introduce and study two new types of bisimulations, weak forward and weak backward~bisimulations, which are more general than forward and backward bisimulations and determine two types of structural equivalence which are closer to the language-equivalence than the FB- and BB-equi\-valence.~We give a way to decide whether there is a weak forward bisimu\-lation between two automata, and if it exists, we provide a way to construct the greatest one (Theorem \ref{th:gwfb}).~Given two automata $\cal A$ and~$\cal B$ and a uniform relation $\varphi \subseteq A\times B $ between their sets of states, we show that $\varphi$ is a weak forward bisimulation if and only if both $E_A^\varphi $ and $E_B^\varphi $ are weak forward bisimulation~equi\-valences on $\cal A$ and $\cal B$, and $\widetilde \varphi $ is a weak forward isomorphism between factor automata ${\cal A}/E_A^\varphi $ and ${\cal B}/E_B^\varphi $ (Theorem \ref{th:uwfb}).~We also characterize uniform weak forward bisimulations between automata $\cal A$ and $\cal B$ in terms of isomorphism between the reverse Nerode automata of $\cal A$ and $\cal B$ (Theorem \ref{th:uwfb.N}).~Finally, we study weak forward bisimulation equivalence between automata and we give an example of automata which are weak forward bisimulation equivalent but not forward bisimulation equivalent.~It should be noted that our concepts of a weak forward bisimulation and~a weak backward bisimulation differ from the concept of a weak bisimulation studied in the concurrency theory.

The paper is organized as follows.~In Section 2 we give definitions of basic notions and notation~concerning
relations and relational calculus, in Section 3 we talk about uniform relations, and in Section 4 we define basic notions and notation concerning nondeterministic automata, introduce factor automata and prove some of their fundamental properties.~In Section 5 we define two types of simulations and four types of bisimulations and discuss the
main properties of forward and backward-forward bisimulations, and in Section 6 we give procedures for
deciding whether there are forward and backward-forward~bisimulations between given automata, and whenever they exist, our procedures compute the greatest~ones.~Section 7 provides characterization results for uniform forward bisimulations, and in Section 8 we define FB-equivalence between automata and prove the main characterization result for FB-equivalent automata.~Section 9 discuss basic properties of backward-forward bisimulations and points to similarities and fundamental differences between them and forward bisimulations.~Then in Section 10 we introduce weak forward and weak backward bisimulations and explore some of their general properties.~In Section 11 we deal with uniform weak forward bisimulations, and in Section 12 we study WFB-equivalence of automata.

It is worth noting that a comprehensive overview of various concepts on deterministic, nondeterministic, fuzzy, and weighted automata, which are related to bisimulations, as well as to the algebraic concepts of a homomorphism, congruence, and relational morphism was given in the penultimate section of \cite{CIDB.11}. It was shown that
all these concepts amount either to forward or to backward-forward bisimulations.

\section{Preliminaries}\label{sec:prel}

Let $A$ and $B$ be non-empty sets.~Any subset $R\subseteq A\times B$ is called a {\it relation
from\/} $A$ {\it to\/} $B$, and~equality, inclusion, union and intersection of relations from $A$ to $B$
are defined as for subsets of $A\times B$. The~{\it inverse\/} of a relation $R\subseteq A\times B$ is
a relation $R^{-1} \subseteq B\times A$~defined by $(b,a)\in R^{-1}$ if and only if $(a,b)\in R$, for all
$a\in A$ and $b\in B$.~If $A=B$, that is, if $R\subseteq A\times A$, then $R$ is called a {\it relation on\/}
$A$.~For a relation $\varphi \subseteq A\times B$ we define a subset $\dom\varphi $ of $A$ and $\im\varphi $ of $B$ by $\dom\varphi =\{a\in A\mid (\exists b\in B)\,(a,b)\in \varphi\}$ and $\im\varphi =\{b\in B\mid (\exists a\in A)\,(a,b)\in \varphi\}$.~We call $\dom\varphi $ the {\it domain\/} of $\varphi $ and $\im\varphi $ the {\it image\/} of $\varphi $.

For non-empty sets $A$, $B$ and $C$, and relations $R\subseteq A\times B$ and $S\subseteq B\times C$,
the {\it composition\/} of $R$ and $S$ is a relation $R\circ S\subseteq A\times C$ defined by
\begin{equation}\label{eq:rel.comp}
(a,c)\in (R \circ S )\ \iff\ (\exists b\in B)\, \bigl((a,b)\in R \land (b,c)\in S\bigr),
\end{equation}
for all $a\in A$ and $c\in C$.~For non-empty sets $A$ and $B$, a relation $R\subseteq A\times B$,
and subsets $\alpha\subseteq A$ and $\beta\subseteq B$, we define subsets $\alpha\circ R\subseteq B$ and
$R\circ \beta \subseteq A$ by
\begin{equation}\label{eq:sr.comp}
b\in \alpha \circ R\ \iff\ (\exists a\in A)\,\bigl(a\in \alpha \land (a,b)\in R\bigr), \ \ \ \
a\in R\circ \beta \ \iff\ (\exists b\in B)\,\bigl( (a,b)\in R \land b\in \beta\bigr),
\end{equation}
for all $a\in A$ and $b\in B$. To simplify our notation, for a non-empty set $A$ and subsets
$\alpha,\beta \subseteq A$ we will write
\begin{equation}\label{eq:ss.comp}
\alpha \circ \beta =\begin{cases}
1 & \text{if $\alpha\cap \beta \ne \emptyset$}, \\ 0 & \text{if $\alpha\cap \beta = \emptyset$}, \\
\end{cases}
\end{equation}
i.e., $\alpha \circ\beta $ is the truth value of the statement "$\alpha\cap \beta \ne \emptyset$".

For non-empty sets $A$, $B$, $C$ and $D$, arbitrary relations $R\subseteq A\times B$, $S,S_1,S_2,S_i\subseteq B\times C$, where $i\in I$,~and $T\subseteq C\times D$, and arbitrary arbitrary subsets $\alpha\subseteq A$, $\beta\subseteq B$, and $\gamma\subseteq C$, the following is true:
\begin{alignat}{2}
& (R\circ S)\circ T = R\circ (S\circ T), \label{eq:comp.as} &&\\
& S_1\subseteq S_2\ \ \text{implies}\ \ R\circ S_1 \subseteq  R\circ S_2 \ \ \text{and}\ \ S_1\circ T \subseteq  S_2\circ T, \label{eq:comp.mon} &&\\
& R\circ \bigl(\bigcup_{i\in I}S_i\bigr) = \bigcup_{i\in I}(R\circ S_i) , \hspace{10mm}
\bigl(\bigcup_{i\in I}S_i\bigr)\circ T = \bigcup_{i\in I}(S_i\circ T) \label{eq:comp.un}&& \\
&(\alpha \circ R)\circ S=\alpha\circ (R\circ S), \hspace{10mm}
(\alpha\circ R)\circ \beta=\alpha\circ (R\circ \beta), \hspace{10mm}
(R\circ S)\circ \gamma=R\circ (S\circ \gamma),\label{eq:comp.sr}&& \\
&(R\circ S)^{-1}=S^{-1}\circ R^{-1}, \label{eq:comp.inv} &&\\
& S_1\subseteq S_2\ \ \text{implies}\ \ S_1^{-1}\subseteq S_2^{-1}, \label{eq:inv.mon} &&\\
&\alpha\circ R=R^{-1}\circ \alpha, \hspace{30mm} R\circ \beta=\beta\circ R^{-1}. \label{eq:inv.comp.rs}&&
\end{alignat}
Therefore, parentheses~in (\ref{eq:comp.as}) and (\ref{eq:comp.sr}) can be omitted.

Note that, despite the notation, the inverse relation $R^{-1}$ is not an inverse of the relation $R\subseteq A\times B$ in~the sense of composition of relations, i.e., $R\circ R^{-1}$ and $R^{-1}\circ R$ are not the equality relations on $A$ and $B$ in general. Let us also note that if $A$, $B$ and $C$ are finite sets with $|A|=k$, $|B|=m$ and $|C|=n$, then $R$ and $S$ can~be treated as $k\times m$ and $m\times n$ Boolean matrices, and $R\circ S$ is their matrix~pro\-duct. Moreover, if we consider
$\alpha $ and $\beta $ as $1\times k$ and $1\times m$ Boolean matrices, i.e., Boolean vectors of length $k$ and $m$, then $\alpha\circ R$ can be treated as the matrix product of $\alpha $ and $R$, $R\circ \beta $ as the matrix product of $R$ and  $\beta^t$ (the transpose~of~$\beta$), and $\alpha \circ \beta $ as the scalar product of vectors $\alpha $ and $\beta$.

Recall that an {\it equivalence\/} on a set $A$ is any reflexive, symmetric and transitive relation on $A$.~Let $E$ be an equivalence on a set $A$. By $E_a$ we denote the equivalence class of an element $a\in A$ with respect to $E$, i.e., $E_a=\{b\in A\mid (a,b)\in E\}$.~The set of all equivalence classes of $E$ is denoted by $A/E$ and called the {\it  factor set\/} of $A$ with respect to $E$.~By $E^\natural $ we denote the {\it natural function\/} of $A$ onto $A/E$,
i.e., the function given by $E^\natural (a)=E_a$, for every $a\in A$.

\section{Uniform relations}

Let $A$ and $B$ be non-empty sets. A relation $\varphi \subseteq A\times B$ is called {\it
complete\/} if for any $a\in A$ there exists $b\in B$ such that $(a,b)\in \varphi $, and {\it
surjective\/} if for any $b\in B$ there exists $a\in A$ such that $(a,b)\in \varphi $. Let us note
that $\varphi $ is complete if and only if there exists a function $f:A\to B$ such that
$(a,f(a))\in \varphi $, for every $a\in A$. Let us call a function $f$ with this property a
{\it functional description\/} of $\varphi $, and let us denote by $FD(\varphi )$ the set of all
such functions. For an equivalence $F$ on $B$, a function $f:A\to B$ is called $F$-{\it
surjective\/}~if~for every $b\in B$ there exists $a\in A$ such that $(f(a),b)\in F$.~In other
words, we have that $f$ is $F$-surjective if and only if $f\circ F^\natural :A\to B/F$ is a surjective
function.

For an arbitrary relation $\varphi \subseteq A\times B$ we define equivalences $E_A^\varphi $ on $A$ and $E_B^\varphi $ on~$B$ in the following way: for all $a_1,a_2\in A$
and $b_1,b_2\in B$ we set
\begin{eqnarray}%
\label{eq:a.phi.1}
(a_1,a_2)\in E_A^\varphi \ &\iff\ (\forall b\in B)(\,(a_1,b)\in \varphi\iff (a_2,b)\in
\varphi \,), \\%
\label{eq:b.phi.1}
(b_1,b_2)\in E_B^\varphi \ &\iff\ (\forall a\in A)(\,(a,b_1)\in \varphi\iff (a,b_2)\in
\varphi \, ).
\end{eqnarray}
We call $E_A^\varphi $ the {\it kernel\/}, and $E_B^\varphi $ the {\it cokernel\/}
of $\varphi $.

Let $A$ and $B$ be non-empty sets.~A {\it partial uniform relation\/} from $A$ to $B$ is a relation
$\varphi \subseteq A\times B$ which satisfies $\varphi\circ \varphi^{-1}\circ \varphi \subseteq \varphi $.~Since
the opposite inclusion always holds, $\varphi $ is a partial uniform relation if and only if $\varphi\circ \varphi^{-1}\circ \varphi = \varphi $.~A partial
uniform relation which is complete and surjective is called a {\it uniform relation\/}.~Let us notice that a partial uniform relation $\varphi \subseteq A\times B$ is a uniform relation from $A'$ to $B'$, where
$A'=\{a\in A\mid (\exists b\in B)\, (a,b)\in\varphi \}$ (the {\it domain\/} of $\varphi $) and
$B'=\{b\in B\mid (\exists a\in A)\, (a,b)\in\varphi \}$ (the {\it image\/} of $\varphi $).

Partial uniform relations and uniform relations are crisp analogues of partial fuzzy
functions and uniform fuzzy relations, which were studied in \cite{CIB.09,ICB.09,Klawonn.00}.~The next two theorems can be derived
from~more general theorems proved in the fuzzy framework (Theorems 3.1 and 3.3 \cite{CIB.09}), but for the sake
of completeness here we give another immediate proofs.

\begin{theorem}\label{th:pur} Let $A$ and $B$ be non-empty sets and let $\varphi \subseteq A\times B$
be a relation.~Then the~following~conditions are equivalent:
\begin{itemize}\parskip0pt
\item[{\rm (i)}] $\varphi $ is a partial uniform relation;
\item[{\rm (ii)}] $\varphi^{-1}$ is a partial uniform relation;
\item[{\rm (iii)}] $\varphi\circ \varphi^{-1}\subseteq E_A^\varphi $;
\item[{\rm (iv)}] $\varphi^{-1}\circ \varphi\subseteq E_B^\varphi $.
\end{itemize}
\end{theorem}

\begin{proof}
(i)$\implies $(iii). Let $(a_1,a_2)\in \varphi\circ \varphi^{-1}$.~Then $(a_1,b_0)\in \varphi $ and
$(b_0,a_2)\in \varphi^{-1}$, for some $b_0\in B$, and for every $b\in B$ we have that $(a_1,b)\in \varphi $
implies $(a_2,b)\in \varphi \circ \varphi^{-1}\circ \varphi \subseteq\varphi $, and likewise,
$(a_2,b)\in \varphi $ implies $(a_1,b)\in \varphi $. Thus, $(a_1,a_2)\in E_A^\varphi $.

(iii)$\implies $(i). Let $(a,b)\in \varphi\circ\varphi^{-1}\circ\varphi $. Then there exist $a'\in
A$ and $b'\in B$ such that  $(a,b')\in \varphi$, $(b',a')\in \varphi^{-1}$ and $(a',b)\in \varphi
$, whence $(a,a')\in \varphi\circ\varphi^{-1}\subseteq E_A^\varphi $, and by $(a',b)\in \varphi
$ and (\ref{eq:a.phi.1}) we obtain $(a,b)\in \varphi $.~Therefore, $\varphi\circ\varphi^{-1}\circ\varphi
\subseteq \varphi$.

Similarly we prove (i)$\iff $(iv), whereas equivalence (i)$\iff $(ii) is obvious.
\end{proof}

If $\varphi \subseteq A\times B$ is a partial uniform relation, then it can be easily verified
that $\varphi\circ \varphi^{-1}$ and $\varphi^{-1}\circ \varphi $ are symmetric and transitive relations,
but they are not necessary reflexive.~Namely, $\varphi\circ \varphi^{-1}$ is reflexive if and only if
$\varphi $ is complete, and $\varphi^{-1}\circ \varphi $ is reflexive if and only if
$\varphi $ is surjective.~Therefore, if $\varphi $ is a uniform relation, then both
$\varphi\circ \varphi^{-1}$ and $\varphi^{-1}\circ \varphi $ are equivalence relations. Moreover,
the following is true.

\begin{theorem}\label{th:ur} Let $A$ and $B$ be non-empty sets and let $\varphi \subseteq A\times B$
be a relation.~Then the~following~conditions are equivalent:
\begin{itemize}\parskip0pt
\item[{\rm (i)}] $\varphi $ is a uniform relation;
\item[{\rm (ii)}] $\varphi^{-1}$ is a uniform relation;
\item[{\rm (iii)}] $\varphi $ is surjective and
$\varphi\circ \varphi^{-1}=E_A^\varphi $;
\item[{\rm (iv)}] $\varphi $ is complete and $\varphi^{-1}\circ \varphi=E_B^\varphi $;
\item[{\rm (v)}] $\varphi $ is complete and for all $f\in FD(\varphi )$, $a\in A$ and $b\in B$, $f$ is $E_B^\varphi $-surjective~and
\begin{equation}\label{eq:phi.EB}
(a,b)\in\varphi\ \iff\ (f(a),b)\in E_B^\varphi ;
\end{equation}
\item[{\rm (vi)}] $\varphi $ is complete and for all $f\in FD(\varphi )$ and $a_1,a_2\in A$, $f$ is $E_B^\varphi $-surjective and
\begin{equation}\label{eq:phi.EA}
(a_1,f(a_2))\in \varphi \ \iff\ (a_1,a_2)\in E_A^\varphi .
\end{equation}
\end{itemize}
\end{theorem}

\begin{proof}
(i)$\iff $(ii). This equivalence is obvious.

(i)$\implies $(iii). According to Theorem \ref{th:pur}, we have that $\varphi\circ\varphi^{-1}\subseteq E_A^\varphi $.

Let $(a_1,a_2)\in E_A^\varphi $. Since $\varphi $ is complete, there exists
$b\in B$ such that $(a_1,b)\in \varphi $, and (\ref{eq:a.phi.1}) yields $(a_2,b)\in \varphi $, so we
obtain that $(a_1,a_2)\in \varphi \circ\varphi^{-1}$.~Therefore, $E_A^\varphi \subseteq \varphi\circ\varphi^{-1}$.

(iii)$\implies $(i). By Theorem \ref{th:pur}, $\varphi $ is a partial uniform relation, by the assumption we have
that it is surjective, and by reflexivity of $\varphi\circ \varphi^{-1}$ it follows that it is complete.

(ii)$\iff $(iv). This equivalence can be proved in the same way as (i)$\iff $(iii).

(iv)$\implies $(v). Let $f\in FD(\varphi)$, $a\in A$ and $b\in B$. If $(a,b)\in \varphi $, then
by this and by $(a,f(a))\in\varphi $ it follows $(f(a),b)\in\varphi^{-1}\circ\varphi
=E_B^\varphi $. On the other hand, if $(f(a),b)\in E_B^\varphi =\varphi^{-1}\circ\varphi $, then
by this and by $(a,f(a))\in\varphi $ it follows $(a,b)\in \varphi\circ\varphi^{-1}\circ\varphi
=\varphi $. Therefore, (\ref{eq:phi.EB}) holds. By (\ref{eq:phi.EB}) and the surjectivity of
$\varphi $ it also follows that $f$ is $E_B^\varphi $-surjective.

(v)$\implies $(iv). By $E_B^\varphi $-surjectivity of $f$ and (\ref{eq:phi.EB}) we obtain that
$\varphi $ is surjective. Let $(b_1,b_2)\in E_B^\varphi $. Then there exists $a\in A$ such that
$(f(a),b_1)\in E_B^\varphi $, and then $(f(a),b_2)\in E_B^\varphi $. Now by (\ref{eq:phi.EB})
it follows that $(a,b_1)\in \varphi $ and $(a,b_2)\in \varphi $, which yields
$(b_1,b_2)\in\varphi^{-1}\circ\varphi $.

Conversely, let $(b_1,b_2)\in \varphi^{-1}\circ \varphi $.~Then there exists $a\in A$ such that
$(a,b_1)\in\varphi $ and $(a,b_2)\in \varphi $, and~by~(\ref{eq:phi.EB}) we obtain that
$(f(a),b_1)\in E_B^\varphi $ and $(f(a),b_2)\in E_B^\varphi $, so $(b_1,b_2)\in E_B^\varphi
$.

(iii)$\iff $(vi). This equivalence can be proved similarly as (iv)$\iff $(v).
\end{proof}

\begin{remark}\rm
Let $A$ and $B$ be non-empty sets and let $\varphi $ be a partial uniform relation from $A$ to $B$.~Then $\varphi $ is a uniform relation from $\dom\varphi $ to $\im\varphi $, and for that reason we introduced the name partial uniform  relation.
\end{remark}

It is easy to check that every equivalence relation and every surjective function are uniform relations, and every function is a partially uniform relation.~This confirms our remark given in the introduction that uniform relations are common generalization of (surjective) functions and equivalence relations.

\begin{theorem}\label{th:ur2}
Let $A$ and $B$ be  non-empty sets, let $E$ be an equivalence on $A$ and $F$ an equivalence on $B$.
Then there exists a uniform relation $\varphi \subseteq A\times B$~such~that $E=E_A^\varphi $ and
$F=E_B^\varphi $ if and only if there exists a bijective function $\phi:A/E\to B/F$.

This bijective function can be represented as $\phi =\widetilde\varphi $, where
 $\widetilde\varphi :A/E\to B/F$~is a function given by
\begin{equation}\label{eq:tphi}
\widetilde\varphi (E_a)= F_{f(a)}, \ \ \text{for any $a\in A$ and $f\in FD(\varphi )$.}
\end{equation}
We also have that $(\widetilde\varphi)^{-1}=\widetilde{\varphi^{-1}}$.
\end{theorem}

\begin{proof}
Let $\varphi \subseteq A\times B$ be a uniform relation~such~that $E=E_A^\varphi $ and
$F=E_B^\varphi $.

First we show that $\widetilde\varphi :A/E\to B/F$ given by (\ref{eq:tphi}) is a well-defined
function, i.e., that it does not depend on the choice of $f\in FD(\varphi )$ and $a\in A$.
Indeed, according to (\ref{eq:phi.EB}) and (\ref{eq:phi.EA}), for any $a_1,a_2\in A$ and
$f_1,f_2\in FD(\varphi)$ we have that
\[
E_{a_1}=E_{a_2} \ \iff\ (a_1,a_2)\in E\ \iff\ (a_1,f_2(a_2))\in \varphi\ \iff\
(f_1(a_1),f_2(a_2))\in F \ \iff\ F_{f_1(a_1)}=F_{f_2(a_2)}.
\]
By this it follows that $\widetilde\varphi $ is well-defined, and also, that it is injective.~Next,
by Theorem \ref{th:ur} (v) and (vi), each $f\in FD(\varphi)$ is $F$-surjective, so we have that
$\widetilde\varphi $ is surjective. Therefore, $\widetilde\varphi $ is a bijective function.

Conversely, let $\phi:A/E\to B/F$ be a bijective function. Let us define $\varphi \subseteq A\times
B$ by
\begin{equation}\label{eq:varphi.phi}
(a,b)\in \varphi \ \iff\ \phi(E_a)=F_b, \ \ \text{for all $a\in A$ and $b\in B$}.
\end{equation}
It is clear that $\varphi $ is complete and surjective.~If $(a,b)\in
\varphi\circ\varphi^{-1}\circ\varphi $, then $(a,b'),(a',b'),(a',b)\in \varphi$, for some $a'\in A$
and $b'\in B$, so $\phi(E_e)=F_{b'}=\phi(E_{a'})=F_b$, whence $(a,b)\in \varphi $. Thus,
$\varphi\circ\varphi^{-1}\circ\varphi \subseteq \varphi $, and since the opposite inclusion is
evident, we conclude that $\varphi $ is a uniform relation.

Next, according to (\ref{eq:a.phi.1}), for arbitrary $a_1,a_2\in A$ we have that
\[
\begin{aligned}
(a_1,a_2)\in E_A^\varphi \ &\iff\ (\forall b\in B)\ \bigl((a_1,b)\in \varphi \iff
(a_2,b)\in\varphi \bigr)\ \iff\ (\forall b\in B)\ \phi(E_{a_1})=F_b \iff \phi(E_{a_2})=F_b \\
&\iff\ \phi(E_{a_1})=\phi(E_{a_2}) \ \iff\ E_{a_1}=E_{a_2} \ \iff \ (a_1,a_2)\in E,
\end{aligned}
\]
and therefore, $E_A^\varphi =E$. Likewise, $E_B^\varphi =F$.

Finally, for every $a\in A$ and $f\in FD(\varphi)$, by $(a,f(a))\in \varphi $ and
(\ref{eq:varphi.phi}) it follows that $\phi (E_a)=F_{f(a)}=\widetilde\varphi(E_a)$, so
$\phi=\widetilde\varphi $. It can be easily verified that
$(\widetilde\varphi)^{-1}=\widetilde{\varphi^{-1}}$.
\end{proof}

Let us note that the bijective function $\widetilde\varphi $ from Theorem \ref{th:ur2} determines some kind of
``uniformity'' between partitions which correspond to the equivalences $E$ and $F$, for what reason we use the name uniform
relation.

\section{Nondeterministic automata and factor automata}

Throughout this paper, if not noted otherwise, let $X$ be a finite non-empty set, called an {\it alphabet\/} (or an {\it input alphabet\/}).~We define a {\it nondeterministic automaton\/} over the alphabet $X$ as a quadruple ${\cal A}=(A,\delta^A,\sigma^A,\tau^A)$, where $A$ is a non-empty set, called the {\it set of states\/}, $\delta^A\subseteq A\times X\times A$ is a ternary
relation, called the {\it transition relation\/}, and $\sigma^A$ and $\tau^A$ are subsets of $A$,~called
respectively the sets of {\it initial states\/} and {\it terminal states\/}.~For~each $x\in X$, a binary relation $\delta_x^A\subseteq A\times A$ defined~by
\[
(a,b)\in \delta_x^A \ \iff\ (a,x,b)\in \delta^A, \ \ \text{for all}\ a,b\in A,
\]
is also called the {\it transition relation\/}.~For any word $u\in X^*$, where $X^*$ is the free monoid over $X$,  the extended transition relation $\delta_u^A\subseteq A\times A$ is defined inductively as follows: for the empty word $\varepsilon\in X^*$ we define $\delta_\varepsilon^A$ to be the equality relation, and for all $u,v\in X^*$ we set $\delta_{uv}^A=\delta_u^A\circ \delta_v^A$.~If we
disregard initial and terminal~states, then the pair ${\cal A}=(A,\delta^A)$ is called a
{\it labelled transition system\/} over $X$ (cf.~\cite{AILS.07,Milner.99}).~Typically, the set of states and the input
alphabet of a nondeterministic automaton are assumed to be finite.~Such assumption is not necessary here, and we will
assume that the input alphabet is finite, but from the methodological reasons, in some cases we will allow the set of states to be infinite.~A nondeterministic automaton whose set of states is finite will be called a {\it nondeterministic finite automaton\/}.~If $\sigma^A=\{a_0\}$, for some $a_0\in A$, and the relation $\delta^A$ is a function from $A\times X$ to $A$, i.e., for every $(a,x)\in A\times X$ there is a unique $a'\in A$ such that $(a,x,a')\in \delta^A$, then $\cal A$ is called a {\it deterministic automaton\/}, and we write ${\cal A}=(A,\delta^A,a_0,\tau^A)$.~In this case, the expressions $(a,x,a')\in \delta^A$ and $\delta^A(a,x)=a'$ will have~the~same meaning.~We also have that $\delta^A_u$ is a function from $A$ to $A$, for every $u\in X^*$, and we will often write $\delta^A_u(a)=a'$ instead of $(a,a')\in \delta^A_u$.~For the sake of simplicity, in the rest of the paper we will say just {\it automaton\/} instead of nondeterministic automaton.

The~{\it reverse auto\-maton\/} of an automaton ${\cal A}=(A,\delta^A,\sigma^A,\tau^A)$
is an automaton $\bar {\cal A}=(A,\bar\delta^{A},\bar\sigma^A,\bar\tau^A )$~whose~transition relation and sets of initial and terminal states are~defined by $\bar\delta^{A} (a,x,b)=\delta
^{A}(b,x,a)$, for all $a,b\in A$ and $x\in X$, $\bar\sigma^A =\tau^A$ and $\bar\tau^A =\sigma^A$. In other words, $\bar\delta_x^A=(\delta_x^A)^{-1}$, for every $x\in X$.

An automaton ${\cal B}=(B,\delta^B,\sigma^B,\tau^B)$ is a {\it subatomaton\/} of an automaton ${\cal A}=(A,\delta^A,\sigma^A,\tau^A)$ if $B\subseteq A$, $\delta_x^B$ is the restriction of $\delta_x^A$ to $B\times B$, for each $x\in X$, and $\sigma^B$ and $\tau^B$ are restrictions of $\sigma^A$ and $\tau^A$ to $B$, i.e., $\delta_x^B=\delta_x^A\cap B\times B$, $\sigma^B=\sigma^A\cap B$, and $\tau^B=\tau^A\cap B$.

Let ${\cal A}=(A,\delta^A,\sigma^A,\tau^A)$ and ${\cal B}=(B,\delta^B,\sigma^B,\tau^B)$ be automata.~A function $\phi :A\to B$ is an {\it isomorphism\/} if~it~is bijective and for all $a,a_1,a_2\in A$ and
$x\in X$ the following is true:
\begin{align}
&(a_1,a_2)\in \delta_x^A \ \iff\ (\phi(a_1),\phi(a_2))\in \delta_x^B, \label{eq:iso1}\\
&a\in \sigma^A\ \iff\ \phi(a)\in \sigma^B, \label{eq:iso2}\\
&a\in \tau^A\ \iff\ \phi(a)\in \tau^B. \label{eq:iso3}
\end{align}
If there exists an isomorphism between $\cal A$ and $\cal B$, then we say that
$\cal A$ and $\cal B$ are {\it isomorphic\/} automata, and we write ${\cal A}\cong
{\cal B}$.~In other words, two automata are isomorphic if in essence they have~the same structure, if they differ eachother only in notation of their states.~In particular, if ${\cal A}=(A,\delta^A,a_0,\tau^A)$ and ${\cal B}=(B,\delta^B,b_0,\tau^B)$ are deterministic automata, then a bijective function $\phi :A\to B$ is an isomorphism if and only if it satisfies $\phi (a_0)=b_0$,  (\ref{eq:iso3}) and
\begin{equation}\label{eq:iso1d}
\phi(\delta^A(a,x))=\delta^B(\phi(a),x),
\end{equation}
for all $x\in x$ and $a\in A$.

It is easy to check that composition of two isomorphisms of automata is also an isomorphism, and thus, for arbitrary automata ${\cal A}$, $\cal B$ and $\cal C$, ${\cal A}\cong {\cal B}$ and ${\cal B}\cong {\cal C}$ implies
${\cal A}\cong {\cal C}$.~A function $\phi :A\to B$ which is injective and it satisfies (\ref{eq:iso1})--(\ref{eq:iso3}) is called a {\it monomorphism\/} from $\cal A$ into $\cal B$.~It is easy to check that $\phi :A\to B$ is a monomorphism from $\cal A$ to $\cal B$ if and only if it is an isomorphism from $\cal A$ to the subautomaton ${\cal C}=(C,\delta^C,\sigma^C,\tau^C)$ of $\cal B$, where $C=\im\phi $.

Let ${\cal A}=(A,\delta^A,\sigma^A,\tau^A)$ be an automaton.~The {\it language recognized by\/} $\cal A$,
denoted~by $L({\cal A})$, is a language in $X^*$ defined as follows: for any $u\in X^*$,
\begin{equation}\label{eq:lang.rec}
u\in L({\cal A})\ \iff\ (\exists a_1,a_2\in A)\, \bigl( a_1\in \sigma^A \land (a_1,a_2)\in \delta_u^A \land a_2\in \tau^A\bigr),
\end{equation}
In notation from Section \ref{sec:prel} (equations (\ref{eq:rel.comp})--(\ref{eq:ss.comp})), the equation (\ref{eq:lang.rec}) ca be also written as
\begin{equation}\label{eq:lang.rec.2}
u\in L({\cal A})\ \iff\ (\sigma^A\circ \delta_u^A) \cap \tau^A\ne \emptyset \ \iff\ \sigma^A\cap (\delta_u^A \circ \tau^A)\ne \emptyset\ \iff\ \sigma^A\circ\delta_u^A\circ\tau^A= 1 .
\end{equation}
Two automata $\cal A$ and $\cal B$ are said to be {\it language-equivalent\/}, or just {\it equivalent\/}, if they recognize the same~language, i.e., if $L({\cal A})=L({\cal B})$.

Let ${\cal A}=(A,\delta^A,\sigma^A,\tau^A)$ be an automaton and let $E$ be an equivalence~on~$A$.~Without any restriction on~the equivalence $E$, we~can~define a transition relation $\delta^{A/E}\subseteq A/E\times X\times A/E$ by
\begin{equation}\label{eq:dE1}
\begin{aligned}
(E_{a_1},x,E_{a_2})\in \delta^{A/E} \ &\iff\ (\exists a_1',a_2'\in A)\, \bigl((a_1,a_1')\in E \land
(a_1',x,a_2')\in \delta^A\land (a_2',a_2)\in E\bigr)\\ &\iff\ (a_1,a_2)\in E\circ \delta_x\circ E,
\end{aligned}
\end{equation}
for all $a_1,a_2\in A$ and $x\in X$, and we can also define sets $\sigma^{A/E},\tau^{A/E}\subseteq A/E$ by
\begin{align}
&E_a\in \sigma^{A/E}\ \iff\ (\exists a'\in A)\,\bigl( a'\in\sigma^A\land (a',a)\in E\bigl) \ \iff\ a\in \sigma^A\circ E, \label{eq:sigmaAE} \\
&E_a\in \tau^{A/E}\ \iff\ (\exists a'\in A)\,\bigl((a,a')\in E\land a'\in\tau^A\bigl) \ \iff\ a\in E\circ \tau^A,\label{eq:tauAE}
\end{align}
for every $a\in A$.~Evidently, $\delta^{A/E}$, $\sigma^{A/E}$ and $\tau^{A/E}$ are well-defined, and
${\cal A}/E=(A/E,\delta^{A/E},\sigma^{A/E},\tau^{A/E})$ is a nondeterministic automaton, called the {\it factor automaton\/} of $\cal A$ w.r.t. $E$.

The next theorems can be conceived as a version, for nondeterministic automata, of two~well-known~theorems from universal algebra: Second Isomorphism Theorem and Correspondence Theorem~(cf.~\cite[II.{\S}6]{BS.81}).

\begin{theorem}\label{th:F:E}
Let ${\cal A}=(A,\delta^A,\sigma^A,\tau^A)$ be an automaton, and let $E$ and $F$ be equivalences
on $A$ such that $E\subseteq F$.

Then a relation $F/E$ on $A/E$ defined by
\begin{equation}\label{eq:F:E}
(E_{a_1},E_{a_2})\in F/E\ \iff\ (a_1,a_2)\in F,\ \ \ \text{for all $a_1,a_2\in A$,}
\end{equation}
is an equivalence on $A/E$, and the factor automata $({\cal A}/E)/(F/E)$ and ${\cal A}/F$ are isomorphic.
\end{theorem}

\begin{proof}
Consider $a_1,a_1',a_2,a_2'\in A$ such that $E_{a_1}=E_{a_1'}$
and $E_{a_2}=E_{a_2'}$, i.e., $(a_1,a_1'),(a_2,a_2')\in E$.~Then~we~have that
$(a_1,a_1'),(a_2,a_2')\in F$, so $(a_1,a_2)\in F$ if and only if $(a_1',a_2')\in F$.~Therefore,
$F/E$ is~a~well-defined relation. It is easy to check that $F/E$ is an equivalence.

For the sake of simplicity set $F/E=P$, and define a function $\phi :A/F\to (A/E)/P$ by
\[
\phi (F_a)=P_{E_a}, \qquad \text{for every $a\in A$}.
\]
For arbitrary $a_1,a_2\in A$  we have that
\[
F_{a_1}=F_{a_2} \ \iff\ \ (a_1,a_2)\in F\ \iff\ (E_{a_1},E_{a_2})\in P \
\iff\  P_{E_{a_1}} = P_{E_{a_2}} \ \iff \ \phi (F_{a_1})=\phi (F_{a_2}),
\]
and hence, $\phi $ is a well-defined and injective function.~It is clear that $\phi
$ is also a surjective function. Therefore, $\phi $ is a bijective function of $A/F$
onto $(A/E)/P$.

Since $E\subseteq F$ is equivalent to $E\circ F=F\circ E=F$, for arbitrary $a_1,a_2\in
A$ and $x\in X$ we have that
\[
\begin{aligned}
&(\phi (F_{a_1}),\phi (F_{a_2}))\in \delta_x^{(A/E)/P}\ \iff\ (P_{E_{a_1}},P_{E_{a_2}})\in \delta_x^{(A/E)/P}\
\iff\ (E_{a_1},E_{a_2})\in (P\circ \delta_x^{A/E}\circ P) \\
&\hspace{30mm}\iff\ (\exists a_3,a_4\in A)\,\bigl((E_{a_1},E_{a_3})\in P \land (E_{a_3},E_{a_4})\in \delta_x^{A/E}
\land (E_{a_4},E_{a_2})\in P\bigr) \\
&\hspace{30mm}\iff\ (\exists a_3,a_4\in A)\,\bigl((a_1,a_3)\in F \land (a_3,a_4)\in (E\circ \delta_x^A\circ E)\land
(a_4,a_2)\in F\bigr)\\
&\hspace{30mm}\iff (a_1,a_2)\in F\circ E\circ \delta_x^A\circ E\circ F = F\circ \delta_x^A\circ F
\\
&\hspace{30mm}\iff\ (F_{a_1},F_{a_2})\in \delta_x^{A/F} .
\end{aligned}
\]
Moreover, for each $a\in A$ we have that
\[
\begin{aligned}
&\phi (F_a)\in \sigma^{(A/E)/P}\ \iff\ P_{E_a}\in \sigma^{(A/E)/P}\ \iff\ E_a\in \sigma^{A/E}\circ P \\
&\hspace{15mm}\iff\ (\exists a'\in A)\,\bigl( E_{a'}\in \sigma^{A/E} \land (E_{a'},E_a)\in P\bigr) \
\iff\ (\exists a'\in A)\,\bigl( a'\in \sigma^A\circ E \land (a',a)\in F\bigr)\\
&\hspace{15mm}\iff a\in \sigma^A\circ E\circ F \ \iff\ a\in \sigma^A\circ F \ \iff\ F_a\in \sigma^{A/F} ,
\end{aligned}
\]
and similarly, $\phi (F_a)\in \tau^{(A/E)/P}\iff F_a\in \tau^{A/F}$.

Hence, $\phi $ is an isomorphism of automata ${\cal A}/F$ and
$({\cal A}/E)/(F/E)$.
\end{proof}

\begin{theorem}\label{th:F:E-isom}
Let ${\cal A}=(A,\delta^A,\sigma^A,\tau^A)$ be an automaton and $E$ an equivalence
on $A$.

The function $\Phi:{\cal E}_E(A)\to {\cal E}(A/E)$, where
${\cal E}_E(A)=\{F\in {\cal E}(A)\mid E\subseteq F\}$, defined by
\begin{equation}\label{eq:corresp}
\Phi (F) = F/E,\ \ \ \text{for every $F\in {\cal E}_E(A)$,}
\end{equation}
is a lattice isomorphism, i.e., it is surjective and
\begin{equation}\label{eq:ord.isom}
F\subseteq G\ \iff\ \Phi (F) \subseteq \Phi(G),\ \ \ \text{for all $F,G\in {\cal E}_E(A)$.}
\end{equation}
\end{theorem}

\begin{proof}
Consider an arbitrary equivalence $P\in {\cal E}(A/E)$. Define a relation $F\subseteq A\times A$ by
\begin{equation}\label{eq:P-F}
(a_1,a_2)\in F\ \iff\ (E_{a_1},E_{a_2})\in P,\ \ \ \text{for all $a_1,a_2\in A$.}
\end{equation}
It is easy to verify that $F$ is an equivalence on $A$, and clearly, $P=F/E$. For arbitrary $a_1,a_2\in A$,
if $(a_1,a_2)\in E$, then $E_{a_1}=E_{a_2}$ and $(E_{a_1},E_{a_2})\in P$, whence it follows that
$(a_1,a_2)\in F$.~Therefore,~$E\subseteq F$,~i.e., $F\in {\cal E}_E(A)$, and we have proved that $\Phi $ is surjective.

Moreover, for arbitrary $F,G\in {\cal E}_E(A)$ we have that
\[
\begin{aligned}
F\subseteq G\ &\iff\ (\forall (a_1,a_2)\in A\times A)\ \bigl((a_1,a_2)\in F\implies (a_1,a_2)\in G\bigr)\\
&\iff\ (\forall (a_1,a_2)\in A\times A)\ \bigl((E_{a_1},E_{a_2})\in \Phi(F)\implies (E_{a_1},E_{a_2})\in
\Phi(G)\bigr)\\
&\iff\ \Phi(F)\subseteq \Phi(G).
\end{aligned}
\]
Therefore, $\Phi $ is a lattice isomorphism.
\end{proof}

It is worth noting that in terms of the lattice theory, ${\cal E}_E(A)$ is the {\it principal filter\/} (or {\it principal dual ideal\/}) of the lattice ${\cal E}(A)$ (which is determined or generated by $E$).

\section{Simulations and bisimulations}

Let ${\cal A}=(A,\delta^A,\sigma^A,\tau^A)$ and ${\cal B}=(B,\delta^B,\sigma^B,\tau^B)$ be automata
and let $\varphi\subseteq A\times B$ be~a non-empty~rela\-tion. We call $\varphi$ a {\it
forward simulation\/} if
\begin{align}
&\sigma^A\subseteq \sigma^B\circ \varphi^{-1}, \label{eq:fs.nda.s} \\
&\varphi^{-1}\circ \delta_x^A\subseteq \delta_x^B\circ \varphi^{-1},\ \ \ \  \text{for every $x\in X$},\label{eq:fs.lts}\\
&\varphi^{-1}\circ \tau^A\subseteq \tau^B , \label{eq:fs.nda.t}
\end{align}
and a {\it backward simulation\/} if
\begin{align}
&\sigma^A\circ \varphi\subseteq \sigma^B, \label{eq:bs.nda.s} \\
&\delta_x^A\circ \varphi \subseteq \varphi\circ \delta_x^B,\ \ \ \  \text{for every $x\in X$}, \label{eq:bs.lts} \\
&\tau^A\subseteq \varphi\circ \tau^B . \label{eq:bs.nda.t}
\end{align}
We call $\varphi $ a {\it forward bisimulation\/} if both $\varphi $ and $\varphi^{-1}$ are
forward simulations,~i.e., if it satisfies (\ref{eq:fs.nda.s})--(\ref{eq:fs.nda.t}) and
\begin{align}
&\sigma^B\subseteq \sigma^A\circ \varphi,\label{eq:fb.nda.si} \\
&\varphi\circ \delta_x^B \subseteq \delta_x^A\circ \varphi,\ \ \text{for every $x\in X$},\label{eq:fb.lts.i}\\
&\varphi\circ \tau^B\subseteq \tau^A, \label{eq:fb.nda.ti}
\end{align}
and a {\it backward bisimulation\/} if both $\varphi $ and $\varphi^{-1}$ are backward simulations,
i.e., if it satisfies (\ref{eq:bs.nda.s})--(\ref{eq:bs.nda.t}) and
\begin{align}
&\sigma^B\circ\varphi^{-1}\subseteq \sigma^A, \label{eq:bb.nda.si}\\
&\delta_x^B\circ \varphi^{-1} \subseteq \varphi^{-1}\circ \delta_x^A,\ \ \ \ \text{for every $x\in X$},\label{eq:bb.lts.i}\\
&\tau^B\subseteq \varphi^{-1}\circ \tau^A . \label{eq:bb.nda.ti}
\end{align}
Let us note that condition (\ref{eq:fs.nda.s}) means that for every $a\in\sigma^A$ there exists $b\in\sigma^B$ such that $(a,b)\in\varphi$, and (\ref{eq:fb.nda.si}) means that for every $b\in\sigma^B$ there exists $a\in\sigma^A$ such that $(a,b)\in\varphi$.~On the other hand, condition (\ref{eq:fs.nda.t}) means that $\{b\in B\mid (\exists a\in \tau^A)\,(a,b)\in\varphi\}\subseteq \tau^B$, and (\ref{eq:fb.nda.ti}) means that $\{a\in A\mid (\exists b\in \tau^B)\,(a,b)\in\varphi\}\subseteq \tau^A$.~Similar interpretations can be given
for conditions (\ref{eq:bs.nda.s}), (\ref{eq:bs.nda.t}), (\ref{eq:bb.nda.si}) and (\ref{eq:bb.nda.ti}).

Next, we call $\varphi$ a {\it forward-backward simulation\/} if $\varphi$ is a forward and $\varphi^{-1}$ is a backward simulation, i.e.,~if
\begin{align}
&\sigma^A=\sigma^B\circ\varphi^{-1} , \label{eq:fbb.nda.s}\\
&\varphi^{-1}\circ \delta_x^A=\delta_x^B\circ \varphi^{-1} ,\ \ \ \ \text{for every $x\in X$},\label{eq:fbb.lts}\\
&\varphi^{-1}\circ \tau^A=\tau^B  , \label{eq:fbb.nda.t}
\end{align}
and a {\it backward-forward simulation\/} if $\varphi$ is a backward and $\varphi^{-1}$ is a forward simulation, i.e.,~if
\begin{align}
&\sigma^A\circ \varphi= \sigma^B, \label{eq:bfb.nda.s} \\
&\delta_x^A\circ \varphi = \varphi\circ \delta_x^B,\ \ \ \  \text{for every $x\in X$}, \label{eq:bfb.lts} \\
&\tau^A= \varphi\circ \tau^B . \label{eq:bfb.nda.t}
\end{align}
For the sake of simplicity, we~will call $\varphi $ just a {\it simulation\/} if it is either a forward or a backward simulation, and just a {\it bisimulation\/} if it is any of the four types of bisimulations defined above.~Moreover, forward~and backward bisimulations will be called {\it homotypic\/}, and backward-forward and forward-backward bisimulations will be called {\it heterotypic\/}.

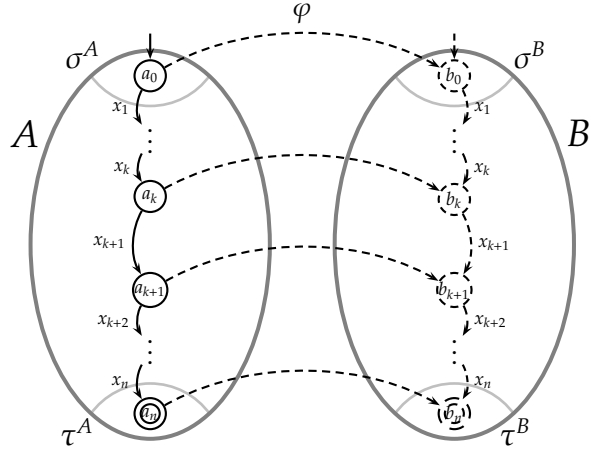
\begin{figure}
\begin{center}
\psset{unit=0.4cm}
\begin{pspicture}(-10,-7.2)(10,5)
\psarc[linecolor=lightgray,linewidth=1pt](-5,6){2.4}{215}{325}
\psarc[linecolor=lightgray,linewidth=1pt](-5,-8){2.4}{35}{145}
\psarc[linecolor=lightgray,linewidth=1pt](5,6){2.4}{215}{325}
\psarc[linecolor=lightgray,linewidth=1pt](5,-8){2.4}{35}{145}
\psellipse[linecolor=gray,linewidth=1.5pt](-5,-1)(4,6.5)
\psellipse[linecolor=gray,linewidth=1.5pt](5,-1)(4,6.5)
\rput(-9.1,2.7){\Large $A$}
\rput(9.1,2.7){\Large $B$}
\rput(-7.2,5.3){\large $\sigma^A$}
\rput(7.5,5.3){\large $\sigma^B$}
\rput(-7.4,-7.2){\large $\tau^A$}
\rput(7.0,-7.2){\large $\tau^B$}
\pnode(-5,6){AI}
\cnode(-5,4.6){.55}{A0}
\rput(A0){\scriptsize $a_0$}
\cnode[linewidth=0,linecolor=white](-5,2.5){.55}{A1}
\rput(-5,2.7){$\vdots$}
\cnode(-5,0.6){.55}{Ak}
\rput(Ak){\scriptsize $a_k$}
\cnode(-5,-2.5){.6}{Ak1}
\rput(Ak1){\scriptsize $a_{k+1}$}
\cnode[linewidth=0,linecolor=white](-5,-4.5){.55}{Ak2}
\rput(-5,-4.3){$\vdots$}
\cnode[doubleline=true,doublesep=1.5pt](-5,-6.6){.55}{An}
\rput(An){\scriptsize $a_n$}
\pnode(5,6){BI}
\cnode[linestyle=dashed,dash=3pt 2pt](5,4.6){.55}{B0}
\rput(B0){\scriptsize $b_0$}
\cnode[linewidth=0,linecolor=white](5,2.5){.55}{B1}
\rput(5,2.7){$\vdots$}
\cnode[linestyle=dashed,dash=3pt 2pt](5,0.5){.55}{Bk}
\rput(Bk){\scriptsize $b_k$}
\cnode[linestyle=dashed,dash=3pt 2pt](5,-2.5){.6}{Bk1}
\rput(Bk1){\scriptsize $b_{k+1}$}
\cnode[linewidth=0,linecolor=white](5,-4.5){.55}{Bk2}
\rput(5,-4.3){$\vdots$}
\cnode[doubleline=true,doublesep=1.5pt,linestyle=dashed,dash=3pt 2pt](5,-6.6){.55}{Bn}
\rput(Bn){\scriptsize $b_n$}
\ncline{->}{AI}{A0}
\ncarc[arcangle=-30]{->}{A0}{A1}\Bput[2pt]{\scriptsize $x_1$}
\ncarc[arcangle=-30]{->}{A1}{Ak}\Bput[2pt]{\scriptsize $x_k$}
\ncarc[arcangle=-30]{->}{Ak}{Ak1}\Bput[2pt]{\scriptsize $x_{k+1}$}
\ncarc[arcangle=-30]{->}{Ak1}{Ak2}\Bput[2pt]{\scriptsize $x_{k+2}$}
\ncarc[arcangle=-30]{->}{Ak2}{An}\Bput[2pt]{\scriptsize $x_n$}
\ncline[linestyle=dashed,dash=3pt 2pt]{->}{BI}{B0}
\ncarc[arcangle=30,linestyle=dashed,dash=3pt 2pt]{->}{B0}{B1}\Aput[2pt]{\scriptsize $x_1$}
\ncarc[arcangle=30,linestyle=dashed,dash=3pt 2pt]{->}{B1}{Bk}\Aput[2pt]{\scriptsize $x_k$}
\ncarc[arcangle=30,linestyle=dashed,dash=3pt 2pt]{->}{Bk}{Bk1}\Aput[2pt]{\scriptsize $x_{k+1}$}
\ncarc[arcangle=30,linestyle=dashed,dash=3pt 2pt]{->}{Bk1}{Bk2}\Aput[2pt]{\scriptsize $x_{k+2}$}
\ncarc[arcangle=30,linestyle=dashed,dash=3pt 2pt]{->}{Bk2}{Bn}\Aput[2pt]{\scriptsize $x_n$}
\ncarc[arcangle=30,linestyle=dashed,dash=3pt 2pt]{->}{A0}{B0}\Aput[3pt]{$\varphi$}
\ncarc[arcangle=30,linestyle=dashed,dash=3pt 2pt]{->}{Ak}{Bk}
\ncarc[arcangle=30,linestyle=dashed,dash=3pt 2pt]{->}{Ak1}{Bk1}
\ncarc[arcangle=30,linestyle=dashed,dash=3pt 2pt]{->}{An}{Bn}
\end{pspicture}\\
\caption{Forward and backward simulation}\label{fig:FBS}
\end{center}
\end{figure}

It is worth to explain the meaning of the names forward and backward simulation.~For this purpose we will use the diagram shown in Figure 1.~Let $\varphi $ be a forward simulation and let $a_0,a_1,\ldots ,a_n$ be an arbitrary successful run of the automaton $\cal A$ on a word $u=x_1x_2\cdots x_n$ ($x_1,x_2,\ldots ,x_n\in X$), i.e., a sequence of states of $\cal A$ such that $a_0\in \sigma^A$,
$(a_k,a_{k+1})\in \delta^A_{x_{k+1}}$, for $0\leqslant k\leqslant n-1$, and $a_n\in \tau^A$.~According to (\ref{eq:fs.nda.s}), there exists an initial state $b_0\in \sigma^B$ such that $(a_0,b_0)\in\varphi $.~Suppose that for some $k$, $0\le k\le n-1$, we have built a sequence of states $b_0,b_1,\ldots ,b_k$ such that
$(b_{i-1},b_i)\in \delta^B_{x_i}$ and $(a_i,b_i)\in \varphi $, for each $i$, $1\le i\le k$.~Then $(b_k,a_{k+1})\in \varphi^{-1}\circ \delta^A_{x_{k+1}}$, and by~(\ref{eq:fs.lts}) we obtain that
$(b_k,a_{k+1})\in \delta^B_{x_{k+1}}\circ \varphi^{-1}$, which means that there exists $b_{k+1}\in B$ such that $(b_k,b_{k+1})\in \delta^B_{x_{k+1}}$ and $(a_{k+1},b_{k+1})\in \varphi $.~Therefore, we have successively built a sequence $b_0,b_1,\ldots ,b_n$ of states of $\cal B$ such that $b_0\in \sigma^B$, $(b_k,b_{k+1})\in \delta^B_{x_{k+1}}$, for every $k$, $0\le k\le n-1$, and $(a_k,b_k)\in \varphi$, for every $k$, $0\le k\le n$.~Moreover, by~(\ref{eq:fs.nda.t}) we obtain that $b_n\in\tau^B$.~Thus, the sequence $b_0,b_1,\ldots ,b_n$ is a successful run of the automaton $\cal B$ on the word $u$ which simulates the original run $a_0,a_1,\ldots ,a_n$ of $\cal A$ on $u$.

In contrast to forward simulations, where we build the sequence $b_0,b_1,\ldots ,b_n$ moving forward, starting with $b_0$ and ending with $b_n$, in the case of backward simulations we build this sequence moving backward, starting with $b_n$ and ending with $b_0$.

In numerous papers dealing with simulations and bisimulations mostly forward simulations and forward~bisimulations have been studied.~They have been usually called just simulations and bisimulations, or {\it strong simulations\/} and {\it strong bisimulations\/} (cf.~\cite{Milner.89,Milner.99,RM-C.00}), and the greatest bisimulation equivalence has been usually called a {\it bisimilarity\/}.~Distinction between forward and backward simulations, and forward and backward bisimulations, has been made, for instance, in \cite{Buch.08,HKPV.98,LV.95} (for various kinds of automata), but~less or more these concepts differ from the concepts having the same name which are considered here.~More~similar to our concepts of forward and backward simulations and bisimulations are those studied in \cite{Brihaye.07}, and in \cite{HMM.07,HMM.09} (for tree automata).

The following lemma can be easily proved by induction.

\begin{lemma}
If condition $(\ref{eq:fs.lts})$ or condition $(\ref{eq:bs.lts})$ holds for every $x\in X$,  then it also holds if we replace the letter $x$ by an arbitrary word $u\in X^*$.
\end{lemma}

We also prove the following two lemmas.

\begin{lemma}\label{le:lang.incl.eq}
Let ${\cal A}=(A,\delta^A,\sigma^A,\tau^A)$ and ${\cal B}=(B,\delta^B,\sigma^B,\tau^B)$ be automata, and let $\varphi \subseteq A\times B$ be a relation.~Then
\begin{itemize}\parskip=0pt
\item[{\rm (a)}] If $\varphi $ is a simulation, then $L({\cal A})\subseteq L({\cal B})$.
\item[{\rm (b)}] If $\varphi $ is a bisimulation, then $L({\cal A})= L({\cal B})$.
\end{itemize}
\end{lemma}

\begin{proof}
(a) Let $\varphi $ be a forward simulation. Then for every $u\in X^*$ we have that
\[
\sigma^A\circ \delta_u^A\circ \tau^A \le \sigma^B\circ \varphi^{-1}\circ\delta_u^A\circ \tau^A \le
\sigma^B\circ\delta_u^B\circ \varphi^{-1}\circ \tau^A \le \sigma^B\circ\delta_u^B\circ \tau^B,
\]
and by (\ref{eq:lang.rec.2}) we obtain that $L({\cal A})\subseteq L({\cal B})$. Similarly, if $\varphi $ is a
backward simulation,~then also $L({\cal A})\subseteq L({\cal B})$.

(b) This follows immediately by (a).
\end{proof}

\begin{lemma}\label{le:dual}
Let ${\cal A}=(A,\delta^A,\sigma^A,\tau^A)$ and ${\cal B}=(B,\delta^B,\sigma^B,\tau^B)$ be automata and let $\varphi\subseteq A\times B$ be a relation.~Then
\begin{itemize}\parskip=0pt
\item[{\rm (a)}] $\varphi $ is a backward bisimulation from $\cal A$ to $\cal B$ if and only if it is a forward bisimulation from $\bar{\cal A}$ to $\bar {\cal B}$.
\item[{\rm (b)}] $\varphi $ is a forward-backward bisimulation from $\cal A$ to $\cal B$ if and only if it is a backward-forward bisimulation from $\bar{\cal A}$ to $\bar {\cal B}$.
\end{itemize}
\end{lemma}

\begin{proof}
It can be easily shown that $\varphi $ is a backward simulation from $\cal A$ to $\cal B$ if and only if $\varphi^{-1}$ is a forward simulation from $\bar {\cal B}$ to $\bar{\cal A}$, and consequently, $\varphi^{-1}$ is a backward simulation from $\cal B$ to $\cal A$ if and only if
$\varphi$ is a forward simulation from $\bar {\cal A}$ to $\bar{\cal B}$.
\end{proof}

According to the previous lemma, for any statement on forward (resp.~backward-forward) bisimulations which~is~univer\-sally~valid (valid for all nondeterministic automata) there is the corresponding universally valid statement on backward (resp.~forward-backward) bisimulations.~For that reason, we will deal only with forward and backward-forward bisimulations.

Let us emphasize the following distinction between homotypic and heterotypic bisimulations.~Evidently, the inverse of a forward (resp.~backward) bisimulation is also a forward (resp.~backward) bisimulation. However, the inverse of a backward-forward (resp.~forward-backward) bisimulation is not necessarily a backward-forward (resp.~forward-backward) bisimulation.~The inverse of a backward-forward bisimulation is a forward-backward bisimulation, and vice versa.~Later we will point out other distinctions.

It is easy to verify that the following is true.

\begin{lemma}\label{le:comp.union}
The composition of two forward {\rm ({\it resp.~backward-forward\/})} bisimulations and the union of an arbitrary family of forward {\rm ({\it resp.~backward-forward\/})} bisimulations are also forward {\rm ({\it resp.~backward-forward\/})} bisimulations.
\end{lemma}

Now we are ready to state and prove the following fundamental result.

\begin{theorem}\label{th:gfb}
Let ${\cal A}=(A,\delta^A,\sigma^A,\tau^A)$ and ${\cal B}=(B,\delta^B,\sigma^B,\tau^B)$ be automata such that there exists at least one forward bisimulation from $\cal A$ to $\cal B$.

Then there exists the greatest forward bisimulation from $\cal A$ to $\cal B$, which is a partial uniform relation.
\end{theorem}

\begin{proof}
By the assumption of the theorem, the family $\{\varphi_i\}_{i\in I}$ of all forward bisimulations from $\cal A$ to $\cal B$ is non-empty.~Let $\varphi $ be the union~of this family.~According to Lemma \ref{le:comp.union}, we obtain that $\varphi $ is a forward bisimulation, and clearly, it is the greatest one.

By Lemma \ref{le:comp.union} we also obtain that $\varphi\circ\varphi^{-1}\circ\varphi $ is a forward bisimulation, and since $\varphi $ is the greatest one, we obtain that $\varphi\circ\varphi^{-1}\circ\varphi \subseteq \varphi $. This means that~$\varphi $ is a partial uniform relation.
\end{proof}

A similar theorem can be proved for backward-forward bisimulations, but there is a difference because in that case we can not prove that the greatest backward-forward bisimulation is a partial uniform relation. In other words, the following is true.

\begin{theorem}\label{th:gbfb}
Let ${\cal A}=(A,\delta^A,\sigma^A,\tau^A)$ and ${\cal B}=(B,\delta^B,\sigma^B,\tau^B)$ be automata such that there exists at least one backward-forward bisimulation from $\cal A$ to $\cal B$.

Then there exists the greatest backward-forward bisimulation from $\cal A$ to $\cal B$.
\end{theorem}

\begin{lemma}\label{le:restr}
Let ${\cal A}=(A,\delta^A,\sigma^A,\tau^A)$ and ${\cal B}=(B,\delta^B,\sigma^B,\tau^B)$ be automata, let $\varphi \subseteq A\times B$ be a relation.~Moreover,~let ${\cal C}= (C,\delta^C,\sigma^C,\tau^C)$ and ${\cal D}= (D,\delta^D,\sigma^D,\tau^D)$ be subautomata of $\cal A$ and $\cal B$, where $C=\dom \varphi $ and $D=\im \varphi$.~Then $\varphi \subseteq C\times D$ and
\begin{itemize}\parskip=0pt
\item[{\rm (a)}] if $\varphi $ is a forward {\rm ({\it resp.~backward\/})} simulation from $\cal A$ to $\cal B$, then it is a forward {\rm ({\it resp.~backward\/})} simulation from $\cal C$ to $\cal D$;
\item[{\rm (b)}] if $\varphi^{-1}$ is a forward {\rm ({\it resp.~backward\/})} simulation from $\cal B$ to $\cal A$, then it is a forward {\rm ({\it resp.~backward\/})} simulation from $\cal D$ to $\cal C$.
\end{itemize}
Also, if $A=C$, then the opposite implication in {\rm (a)} holds, and if $B=D$, then the opposite implication in {\rm (b)} holds.
\end{lemma}

\begin{proof}
We will prove only the part of (a) concerning forward simulations.~The remaining assertions can be proved similarly.~Accordingly, let $\varphi $ be a forward simulation from $\cal A$ to $\cal B$.

First, consider an arbitrary $a\in \sigma^C\subseteq \sigma^A\subseteq \sigma^B\circ \varphi^{-1}$.~Then there exists $b\in B$ such that $b\in \sigma^B$ and $(b,a)\in \varphi^{-1}$, i.e., $(a,b)\in \varphi $, which implies $b\in D$.~This means that $b\in \sigma^B\cap D=\sigma^D$, so $a\in \sigma^D\circ \varphi^{-1}$.~Therefore, we have proved that $\sigma^C\subseteq \sigma^D\circ \varphi^{-1}$.

Next, let $(b,a)\in \varphi^{-1}\circ \delta_x^C\subseteq \varphi^{-1}\circ \delta_x^A\subseteq  \delta_x^B\circ \varphi^{-1}$.~From $(b,a)\in \varphi^{-1}\circ \delta_x^C$ it follows that $(b,a')\in \varphi^{-1}$ and $(a',a)\in \delta_x^C$, for some $a'\in C$, which yields $b\in D$.~Moreover, from $(b,a)\in \delta_x^B\circ \varphi^{-1}$ we obtain that there is $b'\in B$ such that $(b,b')\in \delta_x^B$ and $(b',a)\in \varphi^{-1}$, whence $b'\in D$.~Therefore, we have that $b,b'\in D$ and $(b,b')\in \delta_x^B$, so $(b,b')\in \delta_x^D$, and since $(b',a)\in \varphi^{-1}$, we conclude that $(b,a)\in \delta_x^D\circ \varphi^{-1}$.~Hence, $\varphi^{-1}\circ \delta_x^C\subseteq \delta_x^D\circ \varphi^{-1}$.

Finally, let $b\in \varphi^{-1}\circ \tau^C\subseteq \varphi^{-1}\circ \tau^A\subseteq \tau^B$.~From $b \in \varphi^{-1}\circ \tau^C$ it follows that there exists $a\in C$ such that $(b,a)\in \varphi^{-1}$ and $a\in \tau^C$, whence $b\in D$.~Thus, $b\in \tau^B\cap D=\tau^D$, so
we have proved that $\varphi^{-1}\circ \tau^C\subseteq \tau^D$.

If $A=C$ or $B=D$, then the opposite implications in (a) and (b) are immediate consequences of (\ref{eq:comp.mon}).
\end{proof}

Let ${\cal A}=(A,\delta^A,\sigma^A,\tau^A)$ be an arbitrary automaton.~If $\varphi \subseteq A\times A$ is a forward~bisimulation from $\cal A$ into itself, it will be called a {\it forward bisimulation on\/} $\cal A$ (analogously we define {\it backward bisimu\-la\-tions on\/} $\cal A$).~The family of all forward bisimulations on $\cal A$ is non-empty (it contains at least the~equality relation),~and~as~in the proof of Theorem \ref{th:gfb} it can be shown that there is the greatest
forward bisimulation on $\cal A$, which is~an equivalence (cf.~\cite{AILS.07}, \cite{Milner.99}).~Forward bisimulations on $\cal A$ which are equivalences will be called {\it forward~bisimulation equivalences\/} (analogously we define {\it backward bisimu\-la\-tion equivalences\/}).~The set of all forward bisimulation equivalences on $\cal A$ will be denoted by ${\cal E}^{\mathrm{fb}}({\cal A})$.

By symmetry, an equivalence $E$~on~$A$ is a forward bisimulation on $\cal A$ if and~only~if
\begin{align}\label{eq:fb.on}
&E\circ \delta_x^A \subseteq \delta_x^A \circ E ,\ \ \ \ \text{for each}\ x\in X, \\
\label{eq:fb.ont}
&E\circ \tau^A=\tau^A .
\end{align}
It is worth noting that conditions (\ref{eq:fs.nda.s}) and (\ref{eq:fb.nda.si}) are satisfied whenever $A=B$ and $\varphi $ is a reflexive relation on $A$, and hence, whenever  $A=B$ and $\varphi $ is an equivalence on $A$.~According to Theorem~4.1 \cite{CSIP.10} (see also Theorem 1 \cite{CSIP.07}), condition (\ref{eq:fb.on}) is equivalent to
\begin{equation}\label{eq:fb.on.2}
E\circ \delta_x^A \circ E = \delta_x^A \circ E ,\ \ \text{for each}\ x\in X.
\end{equation}
Similarly, an equivalence $E$~on~$A$ is a backward bisimulation on $\cal A$ if and only if
\begin{align}\label{eq:bb.on}
&\delta_x^A \circ E \subseteq E\circ \delta_x^A ,\ \ \text{for each}\ x\in X, \\
\label{eq:bb.ont}
&\sigma^A\circ E=\sigma^A,
\end{align}
and we also have that condition (\ref{eq:bb.on}) is equivalent to
\begin{equation}\label{eq:bb.on.2}
E\circ \delta_x^A \circ E = E\circ \delta_x^A  ,\ \ \text{for each}\ x\in X.
\end{equation}

Forward bisimulation equivalences have been widely studied in the context of labeled transition systems, where they have been very successfully exploited to reduce the number of states.~In particular, many~algo\-rithms have been proposed to compute the greatest forward bisimulation equivalence on a given labeled transition system.~The faster ones are based on the crucial equivalence between the greatest forward bisimulation equivalence and the relational coarsest partition problem (cf.~\cite{DPP.04,GPP.03,KS.90,RT.08,PT.87,Saha.07}).~Forward~and backward bisimulation equivalences on nondeterministic automata have been studied by Ilie, Yu and others \cite{IY.02,IY.03,INY.04,ISY.05},
where they were respectively called {\it right\/} and {\it left invariant equivalences\/} (see also \cite{CSY.05,CC.04}).~In~a~different context, forward bisimulation equivalences were also discussed by Calude~et al.~\cite{CCK.00}, and there they were called {\it well-behaved equivalences\/}.~Both mentioned types of equivalences were used in reduction of the number of states of~nondeterministic automata.

The next theorem can be deduced by Theorem 4.2 \cite{CSIP.10} (or Theorem 2 \cite{CSIP.07}), but we give a different, direct proof.

\begin{theorem}\label{th:lat.fbe}
Let ${\cal A}=(A,\delta^A,\sigma^A,\tau^A)$ be an automaton.

The set ${\cal E}^{\mathrm{fb}}({\cal A})$ of all forward bisimulation equivalences on $\cal A$ forms a complete lattice.~This lattice is a complete join-subsemilattice of the lattice ${\cal E}(A)$ of all equivalences on $A$.
\end{theorem}

\begin{proof}
Since ${\cal E}^{\mathrm{fb}}({\cal A})$ contains the least element of  ${\cal E}(A)$, the equality relation on $A$, it is enough to prove that ${\cal E}^{\mathrm{fb}}({\cal A})$ is a complete join-subsemilattice of ${\cal E}(A)$.

Let $\{E_i\}_{i\in I}$ be an arbitrary non-empty family of forward bisimulation equivalences on $\cal A$, and let $E$ be the join of this family in the lattice ${\cal E}(A)$.~It is well-known that $E$ can be represented as the set-theoretical union of all relations from $\langle E_i\mid i\in I\rangle $, where $\langle E_i\mid i\in I\rangle $ denotes the subsemigroup, generated by the family $\{E_i\}_{i\in I}$, of the semigroup of all binary relations on $A$.~This means that every relation from~$\langle E_i\mid i\in I\rangle $~can be represented as the composition of some finite collection of relations from $\{E_i\}_{i\in I}$, and according to~Lemma~\ref{le:comp.union}, we conclude that every relation from $\langle E_i\mid i\in I\rangle $ is a forward bisimulation, and therefore, $E$ is a forward bisimulation as the union of all these relations.~Hence, $E\in {\cal E}^{\mathrm{fb}}({\cal A})$,what means that ${\cal E}^{\mathrm{fb}}({\cal A})$ is a complete join-subsemilattice of ${\cal A}(A)$.
\end{proof}

\section{Algorithms for computing the greatest bisimulations}

Kozen in \cite{Kozen.97} provided an algorithm that decides whether there is at least one forward bisimulation~between nondeterministic automata, and when there is a forward bisimulation, the same algorithm computes the greatest one.~Here we give another version of this algorithm, and we also provide an analogous algorithm for backward-forward bisimulations.

For non-empty sets $A$ and $B$ and subsets $\eta\subseteq A$ and $\xi \subseteq B$ we define relations $\eta\rightarrow \xi \subseteq A\times B$ and $\eta\leftarrow \xi \subseteq A\times B$ as follows
\begin{align}
\label{eq:rightarrow}
(a,b) \in \eta\rightarrow \xi   \ \ \Leftrightarrow \ \ (\,a\in \eta \, \Rightarrow \, b\in \xi\,) , \\
\label{eq:leftarrow}
(a,b) \in \eta\leftarrow \xi  \ \ \Leftrightarrow \ \ (\,b\in \xi \, \Rightarrow \, a\in \eta\,) ,
\end{align}
for arbitrary $a\in A$ and $b\in B$.~We prove the following.

\begin{lemma}\label{le:arrows}
Let $A$ and $B$ be non-empty sets and let $\eta\subseteq A$ and $\xi \subseteq B$.
\begin{itemize}\parskip=0pt
\item[{\rm (a)}] The set of all solutions to the inequality $\eta \circ \chi \subseteq \xi $, where $\chi $ is an unknown relation between $A$ and $B$, is the principal ideal of ${\cal R}(A,B)$ generated by the relation  $\eta\rightarrow \xi$.
\item[{\rm (b)}] The set of all solutions to the inequality $\chi \circ \xi \subseteq \eta $, where $\chi $ is an unknown relation between $A$ and $B$, is the principal ideal of ${\cal R}(A,B)$ generated by the relation  $\eta\leftarrow \xi$.
\end{itemize}
\end{lemma}

\begin{proof}
(a) Let a relation $\varphi \subseteq A\times B$ be a solution to $\eta \circ \chi \subseteq \xi $, and let $(a,b)\in \varphi $.~If $a\in \eta $, then $b\in \eta\circ \varphi \subseteq \xi $, and according to
(\ref{eq:rightarrow}) we conclude that $(a,b)\in \eta\rightarrow \xi  $.~Thus, $\varphi \subseteq \eta\rightarrow \xi$.

Conversely, assume that $\varphi \subseteq \eta\rightarrow \xi$.~Then for an arbitrary $b\in \eta\circ \varphi $ we have that there exists $a\in \eta $ such that $(a,b)\in \varphi\subseteq \eta\rightarrow \xi$, and again by (\ref{eq:rightarrow}) we conclude that $b\in \xi $.~Hence, $\varphi $ is a solution to $\eta \circ \chi \subseteq \xi $, and consequently, the assertion (a) is true.

The assertion (b) can be proved in a similar way.
\end{proof}

It is worth noting that $(\eta\rightarrow \xi )\cap (\eta\leftarrow \xi )=(\eta\times \xi)\cup ((A\setminus\eta)\times (B\setminus \xi))=\eta\leftrightarrow \xi$, where $\eta\leftrightarrow \xi$ is a relation between $A$ and $B$ defined by
\begin{equation}
\label{eq:leftrightarrow}
(a,b) \in \eta\leftrightarrow \xi   \ \ \Leftrightarrow \ \ (\,a\in \eta \, \Leftrightarrow \, b\in \xi\,) ,
\end{equation}
for arbitrary $a\in A$ and $b\in B$.

Next, let $A$ and $B$ be non-empty sets and let $\alpha \subseteq A\times A$, $\beta\subseteq B\times B$ and $\varphi \subseteq A\times B$.~The {\it right residual\/} of $\varphi $ by $\alpha $ is a relation $\varphi /\alpha \subseteq A\times B$ defined by
\begin{equation}\label{eq:rr.def}
(a,b)\in \varphi /\alpha \ \ \Leftrightarrow \ \ (\forall a'\in A)\, \left(\, (a',a)\in\alpha \, \Rightarrow \, (a',b)\in\varphi\, \right),
\end{equation}
for all $a\in A$ and $b\in B$, and the {\it left residual\/} of $\varphi $ by $\beta $ is a relation $\varphi \backslash\beta \subseteq A\times B$ defined by
\begin{equation}\label{eq:lr.def}
(a,b)\in \varphi \backslash\beta \ \ \Leftrightarrow \ \ (\forall b'\in B)\, \left(\, (b,b')\in\beta \, \Rightarrow \, (a,b')\in\varphi\, \right),
\end{equation}
for all $a\in A$ and $b\in B$.~In the case when $A=B$, these two concepts become the well-known concepts of~right and left residuals of relations on a set (cf.~\cite{Birkhoff,BE.99}).~We have the following.

\begin{lemma}\label{le:residuals}
Let $A$ and $B$ be non-empty sets and let $\alpha \subseteq A\times A$, $\beta\subseteq B\times B$ and $\varphi \subseteq A\times B$.
\begin{itemize}\parskip=0pt
\item[{\rm (a)}] The set of all solutions to the inequality $\alpha \circ \chi \subseteq \varphi $, where $\chi $ is an unknown relation between $A$ and $B$, is the principal ideal of ${\cal R}(A,B)$ generated by the right residual $\varphi /\alpha $ of of $\varphi $ by $\alpha $.
\item[{\rm (b)}] The set of all solutions to the inequality $\chi \circ \beta \subseteq \varphi $, where $\chi $ is an unknown relation between $A$ and $B$, is the principal ideal of ${\cal R}(A,B)$ generated by the left residual $\varphi \backslash\beta $ of of $\varphi $ by $\beta $.
\end{itemize}
\end{lemma}

\begin{proof}
(a) Let $\psi \subseteq A\times B$ be an arbitrary solution to $\alpha \circ \chi \subseteq \varphi $, and let $(a,b)\in \psi $.~For every $a'\in A$, if $(a',a)\in \alpha $, then $(a',b)\in \alpha\circ \psi \subseteq \varphi $, and according to (\ref{eq:rr.def}), we conclude that $(a,b)\in \varphi/\alpha$.~Therefore, $\psi\subseteq \varphi/\alpha $.

On the other hand, let $\psi\subseteq \varphi/\alpha $ and let $(a,b)\in \alpha \circ \psi $. Then there exists $a'\in A$ such that $(a,a')\in \alpha $ and $(a',b)\in \psi\subseteq \varphi/\alpha $, and by (\ref{eq:rr.def}) we obtain that $(a,b)\in \varphi $.~Hence, $\psi $ is a solution to $\alpha \circ \chi \subseteq \varphi $, and consequently, we conclude that (a) is true.

The assertion (b) can be proved analogously.
\end{proof}

We are now ready to state and prove the following theorem, which provides an algorithm that decides whether there is a forward bisimulation between two automata and computes the greatest forward~bisimulation.

\begin{theorem}\label{th:gfb.alg}
Let ${\cal A}=(A,\delta^A,\sigma^A,\tau^A)$ and ${\cal B}=(B,\delta^B,\sigma^B,\tau^B)$ be finite automata.~Define inductively a sequence $\{\varphi_k\}_{k\in \Bbb N}$ of relations between $A$ and $B$ as follows:
\begin{align}\label{eq:phi1}
&\varphi_1= \tau^A\leftrightarrow \tau^B, \\
\label{eq:phi2}
&\varphi_{k+1}=\varphi_k\cap \bigcap_{x\in X}\left(((\delta_x^B\circ\varphi_k^{-1})\setminus \delta_x^A)^{-1}\cap
((\delta_x^A\circ\varphi_k)\setminus \delta_x^B)\right).
\end{align}
Then $\{\varphi_k\}_{k\in \Bbb N}$ is a non-increasing sequence of relations and there exists $k\in \Bbb N$ such that $\varphi_k=\varphi_{k+1}$.

The relation $\varphi_k$ is the greatest relation between $A$ and $B$ which satisfies conditions $(\ref{eq:fs.lts})$, $(\ref{eq:fs.nda.t})$, $(\ref{eq:fb.lts.i})$, and $(\ref{eq:fb.nda.ti})$.~Moreover, if $\varphi_k$ satisfies conditions $(\ref{eq:fs.nda.s})$ and $(\ref{eq:fb.nda.si})$, then $\varphi_k$ is the greatest forward bisimulation between $\cal A$ and $\cal B$, and otherwise, if $\varphi_k$ does not satisfy these conditions, then there is no any forward bisimulation between $\cal A$ and $\cal B$.
\end{theorem}

\begin{proof}
(a) It is clear that $\varphi_{k+1}\subseteq \varphi_k$, for every $k\in \Bbb N$.~As the sets $A$ and $B$ are finite, there is a finite number of relations between $A$ and $B$, so there are $k,m\in \Bbb N$ such that $\varphi_k=\varphi_{k+m}$.~Now, $\varphi_{k+1}\subseteq \varphi_{k+m}= \varphi_k \subseteq \varphi_{k+1}$, and hence,
$\varphi_k=\varphi_{k+1}$.

Next, set $\varphi =\varphi_k$.~Acording to Lemma \ref{le:arrows}, a relation $\psi \subseteq A\times B$ satisfies (\ref{eq:fs.nda.t}) and (\ref{eq:fb.nda.ti}) if and only if $\psi \subseteq \varphi_1$, and hence, $\varphi $ satisfies (\ref{eq:fs.nda.t}) and (\ref{eq:fb.nda.ti}). Furthermore, by (\ref{eq:phi2}) it follows that
\[
\varphi =\varphi \cap \bigcap_{x\in X}\left(((\delta_x^B\circ\varphi^{-1})\setminus \delta_x^A)^{-1}\cap
((\delta_x^A\circ\varphi)\setminus \delta_x^B)\right),
\]
and for every $x\in X$ we obtain that $\varphi \subseteq ((\delta_x^B\circ\varphi^{-1})\setminus \delta_x^A)^{-1}$ and $\varphi \subseteq (\delta_x^A\circ\varphi)\setminus \delta_x^B$, i.e., $\varphi^{-1} \subseteq (\delta_x^B\circ\varphi^{-1})\setminus \delta_x^A$ and $\varphi \subseteq (\delta_x^A\circ\varphi)\setminus \delta_x^B$.~According to (b) of Lemma \ref{le:residuals}, $\varphi^{-1}\circ \delta_x^A \subseteq \delta_x^B\circ\varphi^{-1}$ and $\varphi \circ \delta_x^B \subseteq \delta_x^A\circ\varphi$, and thus, $\varphi $ satisfies conditions (\ref{eq:fs.lts}) and (\ref{eq:fb.lts.i}).

Let $\psi \subseteq A\times B$ be an arbitrary relation satisfying conditions (\ref{eq:fs.lts}), (\ref{eq:fs.nda.t}), (\ref{eq:fb.lts.i}), and (\ref{eq:fb.nda.ti}).~As we have already said, $\psi $ satisfies (\ref{eq:fs.nda.t}) and (\ref{eq:fb.nda.ti}) if and only if $\psi \subseteq \varphi_1$.~Suppose that $\psi \subseteq \varphi_i$, for some $i\in \Bbb N$.~Then for every $x\in X$ we have that $\psi^{-1}\circ \delta_x^A\subseteq \delta_x^B\circ \psi^{-1}\subseteq  \delta_x^B\circ \varphi_i^{-1}$, and according to (b) of Lemma \ref{le:residuals}, $\psi^{-1}\subseteq (\delta_x^B\circ\varphi_i^{-1})\setminus \delta_x^A$, that is, $\psi \subseteq ((\delta_x^B\circ\varphi_i^{-1})\setminus \delta_x^A)^{-1}$.~Analogously we show that $\psi \subseteq (\delta_x^A\circ\varphi_i)\setminus \delta_x^B$.~Therefore,
\[
\psi\subseteq \varphi_i\cap \bigcap_{x\in X}\left(((\delta_x^B\circ\varphi_i^{-1})\setminus \delta_x^A)^{-1}\cap
((\delta_x^A\circ\varphi_i)\setminus \delta_x^B)\right)= \varphi_{i+1}.
\]
Now, by induction we conclude that $\psi \subseteq \varphi_i$, for each $i\in \Bbb N$, and hence, $\psi \subseteq \varphi $.~This means that $\varphi $~is~the~greatest relation satisfying conditions (\ref{eq:fs.lts}), (\ref{eq:fs.nda.t}), (\ref{eq:fb.lts.i}), and (\ref{eq:fb.nda.ti}).

In addition, if $\varphi $ satisfies conditions (\ref{eq:fs.nda.s}) and (\ref{eq:fb.nda.si}), then it is a forward bisimulation between $\cal A$ and $\cal B$, and it is just the greatest one. On the other hand, assume that $\varphi$ does not satisfies (\ref{eq:fs.nda.s}) and (\ref{eq:fb.nda.si}). If $\psi $~is~an~arbitrary forward bisimulation between $\cal A$ and $\cal B$, then it satisfies conditions (\ref{eq:fs.lts}), (\ref{eq:fs.nda.t}), (\ref{eq:fb.lts.i}), and (\ref{eq:fb.nda.ti}), and hence, $\psi \subseteq \varphi $.~From this it follows that $\sigma^A\subseteq \sigma^B\circ \psi^{-1}\subseteq \sigma^B\circ \varphi^{-1}$ and $\sigma^B\subseteq \sigma^A\circ \psi\subseteq \sigma^A\circ \varphi$, which leads to contradiction.~Therefore, we conclude that if $\varphi$ does not satisfy conditions (\ref{eq:fs.nda.s}) and (\ref{eq:fb.nda.si}), then there is no any forward bisimulation between $\cal A$ and $\cal B$.
\end{proof}

Therefore, to decide whether there exists a forward bisimulation between two automata and compute~the greatest one, we build a sequence $\{\varphi_k\}_{k\in \Bbb N}$ of relations in the following way.~The first relation $\varphi_1$ is computed as the greatest relation that satisfies the conditions (\ref{eq:fs.nda.t}) and (\ref{eq:fb.nda.ti}).~Then we start an iterative procedure which computes $\varphi_{k+1}$ from $\varphi_k$ and check whether $\varphi_{k+1}=\varphi_k$.~The procedure terminates when we find the smallest $k\in \Bbb N$ such that $\varphi_{k+1}=\varphi_k$.~After that we check whether $\varphi_k$ satisfies conditions (\ref{eq:fs.nda.s}) and (\ref{eq:fb.nda.si}). If $\varphi_k$ does not satisfy these conditions, we conclude that there is no any forward bisimulation between the given automata, and if $\varphi_k$ satisfies (\ref{eq:fs.nda.s}) and (\ref{eq:fb.nda.si}), we conclude that it is the greatest forward bisimulation between the given automata.

The application of this algorithm is demonstrated by the following example.

\begin{example}\rm
Let ${\cal A}=(A,\delta^A,\sigma^A,\tau^A)$ and ${\cal B}=(B,\delta^B,\sigma^B,\tau^B)$ be
automata with $|A|=3$, $|B|=5$ and $X=\{x,y\}$, whose transition relations and sets of
initial and terminal states are represented by the following Boolean matrices and vectors:
\[
\begin{aligned}
&\delta_x^A=\begin{bmatrix}
1 & 1 & 0 \\
0 & 1 & 1 \\
1 & 0 & 0 \\
\end{bmatrix},\ \ \  \
\delta_y^A=\begin{bmatrix}
1 & 1 & 0 \\
0 & 0 & 1 \\
0 & 0 & 1 \\
\end{bmatrix},\ \ \  \
\delta_x^B=\begin{bmatrix}
1 & 1 & 0 & 1 & 0 \\
1 & 1 & 0 & 1 & 0 \\
1 & 1 & 0 & 0 & 0 \\
0 & 0 & 1 & 1 & 1 \\
1 & 1 & 0 & 0 & 0 \\
\end{bmatrix},\ \ \  \
\delta_y^B=\begin{bmatrix}
1 & 1 & 0 & 1 & 0 \\
1 & 1 & 0 & 1 & 0 \\
0 & 0 & 1 & 0 & 1 \\
0 & 0 & 1 & 0 & 1 \\
0 & 0 & 1 & 0 & 1 \\
\end{bmatrix} \\
&\sigma^A=\begin{bmatrix}
1  &
0  &
0
\end{bmatrix},\ \ \  \
\sigma^B=\begin{bmatrix}
1  &
1  &
0  &
0  &
0
\end{bmatrix},\ \ \  \
\tau^A=\begin{bmatrix}
0  \\
0  \\
1  \\
\end{bmatrix},\ \ \  \
\tau^B=\begin{bmatrix}
0  \\
0  \\
1  \\
0  \\
1  \\
\end{bmatrix}.
\end{aligned}
\]
Using the above described procedure we obtain that
\[
\varphi_{1}=\begin{bmatrix}
1 & 1 & 0 & 1 & 0 \\
1 & 1 & 0 & 1 & 0 \\
0 & 0 & 1 & 0 & 1
\end{bmatrix},\ \ \  \
\varphi_{2}=\varphi_{3}=\begin{bmatrix}
1 & 1 & 0 & 0 & 0 \\
0 & 0 & 0 & 1 & 0 \\
0 & 0 & 1 & 0 & 1
\end{bmatrix}.
\]
It is easy to check that $\varphi_2$ satisfies conditions (\ref{eq:fs.nda.s}) and (\ref{eq:fb.nda.si}), and therefore, $\varphi_2$ is the greatest forward bisimulation between automata $\cal A$ and $\cal B$.
\end{example}

The following theorem, which can be proved in a similar way as Theorem \ref{th:gfb.alg}, provides an algorithm~that decides whether there is a backward-forward bisimulation between two automata and computes the greatest backward-forward~bisimulation.

\begin{theorem}\label{th:gbfb.alg}
Let ${\cal A}=(A,\delta^A,\sigma^A,\tau^A)$ and ${\cal B}=(B,\delta^B,\sigma^B,\tau^B)$ be finite automata.~Define inductively a sequence $\{\varphi_k\}_{k\in \Bbb N}$ of relations between $A$ and $B$ as follows:
\begin{align}\label{eq:phi1}
&\varphi_1= (\sigma^A\rightarrow\sigma^B)\cap (\tau^A\leftarrow\tau^B), \\
\label{eq:phi2}
&\varphi_{k+1}=\varphi_k\cap \bigcap_{x\in X}\left((\delta_x^A\circ\varphi_k)\backslash \delta_x^B)\cap
((\varphi_k\circ\delta_x^B)/ \delta_x^A)\right).
\end{align}
Then $\{\varphi_k\}_{k\in \Bbb N}$ is a non-increasing sequence of relations and there exists $k\in \Bbb N$ such that $\varphi_k=\varphi_{k+1}$.

The relation $\varphi_k$ is the greatest relation between $A$ and $B$ which satisfies conditions $(\ref{eq:bs.nda.s})$, $(\ref{eq:bs.lts})$, $(\ref{eq:fb.lts.i})$, and $(\ref{eq:fb.nda.ti})$.~Moreover,~if $\varphi_k$ satisfies conditions $(\ref{eq:bs.nda.t})$ and $(\ref{eq:fb.nda.si})$, then $\varphi_k$ is the greatest backward-forward bisimulation between $\cal A$ and $\cal B$, and other\-wise, if $\varphi_k$ does not satisfy these conditions, then there is no any backward-forward bisimulation between $\cal A$~and~$\cal B$.
\end{theorem}

The following example shows the case when there is a backward-forward bisimulation, but there is no a forward bisimulation between two automata.

\begin{example}\rm
Let ${\cal A}=(A,\delta^A,\sigma^A,\tau^A)$ and ${\cal B}=(B,\delta^B,\sigma^B,\tau^B)$ be
automata with $|A|=2$, $|B|=3$ and $X=\{x,y\}$, whose transition relations and sets of
initial and terminal states are represented by the following Boolean matrices and vectors:
\[
\begin{aligned}
&\delta_x^A=\begin{bmatrix}
1 & 0 \\
1 & 1 \\
\end{bmatrix},\ \ \  \
\delta_y^A=\begin{bmatrix}
1 & 0 \\
1 & 0 \\
\end{bmatrix},\ \ \  \
\delta_x^B=\begin{bmatrix}
1 & 0 & 0 \\
0 & 1 & 1 \\
0 & 0 & 1 \\
\end{bmatrix},\ \ \  \
\delta_y^B=\begin{bmatrix}
1 & 0 & 1 \\
1 & 0 & 0 \\
0 & 0 & 0\end{bmatrix} \\
&\sigma^{A}=\begin{bmatrix}
1  &
0\end{bmatrix},\ \ \  \
\sigma^B=\begin{bmatrix}
1  &
0  &
1
\end{bmatrix},\ \ \  \
\tau^A=\begin{bmatrix}
0  \\
1  \\
\end{bmatrix},\ \ \  \
\tau^B=\begin{bmatrix}
0  \\
1  \\
0  \\
\end{bmatrix}.
\end{aligned}
\]
Using the procedure from Theorem \ref{th:gbfb.alg} we obtain that
\[
\varphi=\begin{bmatrix}
1 & 0 & 1 \\
1 & 1 & 0 \\
\end{bmatrix},
\]
is the greatest backward-forward bisimulation between $\cal A$ and $\cal B$. On the other hand, using the procedure from Theorem \ref{th:gfb.alg} we obtain that there is no a forward bisimulation between  ${\cal A}$ and ${\cal B}$.

Moreover, it is easy to verify that $\varphi $ is not a partial uniform relation, which confirms our ascertainment given immediately before Theorem \ref{th:gbfb}.
\end{example}

\section{Uniform forward bisimulations}

In this section we deal with forward bisimulations which are uniform relations.~First we show that~within the class of uniform relations forward bisimulations can be characterized by means of equalities.

\begin{theorem}\label{th:ufb0}
Let ${\cal A}=(A,\delta^A,\sigma^A,\tau^A)$ and ${\cal B}=(B,\delta^B,\sigma^B,\tau^B)$ be
automata and let $\varphi \subseteq A\times B$ be a uniform relation.
Then $\varphi $ is a forward bisimulation if and only if the following hold:
\begin{align}
\sigma^A\circ\varphi\circ\varphi^{-1}&=\sigma^B\circ\varphi^{-1}, & \sigma^A\circ\varphi&=\sigma^B\circ\varphi^{-1}\circ\varphi, &&\label{eq:ufb0.3} \\
\delta_x^A\circ \varphi\circ \varphi^{-1}&=\varphi \circ \delta_x^B\circ \varphi^{-1},&
\varphi^{-1}\circ \delta_x^A\circ \varphi&=\delta_x^B\circ \varphi^{-1}\circ \varphi ,&\ \ \ \
&\text{for every $x\in X$}, \label{eq:ufb0.1} \\
\tau^A&=\varphi\circ \tau^B  ,&\ \ \varphi^{-1}\circ \tau^A&= \tau^B.& & \label{eq:ufb0.4}
\end{align}
\end{theorem}

\begin{proof} Let $\varphi $ be a forward bisimulation.~By (\ref{eq:comp.mon}), (\ref{eq:fs.nda.s}), and (\ref{eq:fb.nda.si}), we obtain $\sigma^A\circ\varphi\subseteq \sigma^B\circ\varphi^{-1}\circ\varphi\subseteq
 \sigma^A\circ\varphi $, so $\sigma^A\circ\varphi=\sigma^B\circ\varphi^{-1}\circ\varphi$, and by this it follows that
$\sigma^A\circ\varphi\circ\varphi^{-1}=\sigma^B\circ\varphi^{-1}\circ\varphi\circ\varphi^{-1}=\sigma^B\circ\varphi^{-1}$.

Next, by Theorem \ref{th:ur} and Lemma \ref{le:comp.union} we obtain that $\varphi\circ\varphi^{-1}$ is a forward bisimulation
equivalence on $\cal A$, and according to (\ref{eq:fb.on.2}), for every $x\in X$ we have
\[
\varphi\circ\delta_x^B\circ\varphi^{-1}\subseteq \delta_x^A\circ\varphi\circ\varphi^{-1}=
\varphi\circ\varphi^{-1}\circ\delta_x^A\circ\varphi\circ\varphi^{-1}\subseteq
\varphi\circ\delta_x^B\circ\varphi^{-1}\circ\varphi\circ\varphi^{-1} = \varphi\circ\delta_x^B\circ\varphi^{-1}.
\]
Therefore, $\delta_x^A\circ\varphi\circ\varphi^{-1}=\varphi\circ\delta_x^B\circ\varphi^{-1}$.~In a similar way
we show that $\varphi^{-1}\circ \delta_x^A\circ \varphi = \delta_x^B\circ
\varphi^{-1}\circ \varphi$.

Finally, since $\varphi\circ\varphi^{-1}$ is a forward bisimulation equivalence on $\cal A$, by (\ref{eq:fb.ont}), (\ref{eq:fs.nda.t}), (\ref{eq:comp.mon}), and (\ref{eq:fb.nda.ti}), we obtain that $\tau^A=\varphi\circ\varphi^{-1}\circ\tau^A \subseteq \varphi\circ\tau^B\subseteq \tau^A$, and hence, $\tau^A=\varphi\circ\tau^B$.~Similarly we show that $\varphi^{-1}\circ\tau^A=\tau^B$.

Therefore, we have proved that (\ref{eq:ufb0.3})--(\ref{eq:ufb0.4}) are true.

Conversely, let (\ref{eq:ufb0.3})--(\ref{eq:ufb0.4}) hold.~By the reflexivity of $\varphi\circ \varphi^{-1}$ and (\ref{eq:ufb0.3})
we have $\sigma^A\subseteq \sigma^A\circ\varphi\circ\varphi^{-1}= \sigma^B\circ\varphi^{-1}$, and thus, (\ref{eq:fs.nda.s}) holds.
Furthermore, by the reflexivity of $\varphi\circ \varphi^{-1}$, (\ref{eq:comp.mon}), and (\ref{eq:ufb0.1}), for each $x\in X$ we have that
\[
\varphi^{-1}\circ\delta_x^A\subseteq \varphi^{-1}\circ\delta_x^A\circ\varphi\circ\varphi^{-1}=
\delta_x^B\circ\varphi^{-1}\circ\varphi\circ\varphi^{-1} = \delta_x^B\circ\varphi^{-1},
\]
so $\varphi^{-1}\circ\delta_x^A\subseteq \delta_x^B\circ\varphi^{-1}$, and similarly,
$\varphi\circ\delta_x^B\subseteq \delta_x^A\circ\varphi $.~Finally, it is clear that (\ref{eq:ufb0.4}) implies
(\ref{eq:fs.nda.t}) and (\ref{eq:fb.nda.ti}).~Therefore, we have proved that $\varphi $ is a forward bisimulation.
\end{proof}

Because of the symmetry in (\ref{eq:ufb0.3}) we have included two equalities, although any of them is sufficient,~while the other is unnecessary.~For instance, if
$\sigma^A\circ\varphi\circ\varphi^{-1}=\sigma^B\circ\varphi^{-1}$ then $\sigma^A\circ\varphi =\sigma^A\circ\varphi\circ\varphi^{-1}\circ\varphi = \sigma^B\circ\varphi^{-1}\circ\varphi$, and similarly we show that the second equality implies the first one.

The following theorem is one of the main results of this article.~It gives a characterization of uniform forward bisimulations in terms of the properties of their kernels, cokernels, and related factor automata.

\begin{theorem}\label{th:ufb}
Let ${\cal A}=(A,\delta^A,\sigma^A,\tau^A)$ and ${\cal B}=(B,\delta^B,\sigma^B,\tau^B)$ be automata and
let $\varphi \subseteq A\times B$ be a uniform relation. Then $\varphi $ is a forward bisimulation if and only if
the following hold:
\begin{itemize}\parskip=0pt\itemindent6pt
\item[{\rm (i)}] $E_A^\varphi $ is a forward bisimulation equivalence on $\cal A$;
\item[{\rm (ii)}] $E_B^\varphi $ is a forward bisimulation equivalence on $\cal B$;
\item[{\rm (iii)}]  $\widetilde \varphi $ is an isomorphism of factor automata ${\cal A}/E_A^\varphi $ and ${\cal B}/E_B^\varphi $.
\end{itemize}
\end{theorem}

\begin{proof}
For the sake of simplicity set $E_A^\varphi =E$ and $E_B^\varphi =F$.

Let $\varphi $ be a forward bisimulation. According to Theorem \ref{th:ufb0}, for every $x\in X$ we have that
\[
E\circ \delta_x^A\circ E =\varphi\circ\varphi^{-1}\circ\delta_x^A\circ \varphi\circ\varphi^{-1} =
\varphi\circ\delta_x^B\circ\varphi^{-1}\circ\varphi\circ\varphi^{-1} = \varphi\circ\delta_x^B\circ\varphi^{-1}=
\delta_x^A\circ\varphi\circ\varphi^{-1}=\delta_x^A\circ E,
\]
and also, $E\circ \tau^A = \varphi\circ\varphi^{-1}\circ \tau^A= \varphi\circ \tau^B= \tau^A $.~The inclusion $\sigma^A\subseteq \sigma^A\circ E$ is evident.~Hence, $E=E_A^\varphi $~is a forward bisimulation equivalence~on~$\cal A$. Likewise, $F=E_B^\varphi $ is a forward bisimulation equivalence on $\cal B$.

By Theorem \ref{th:ur2}, $\widetilde\varphi $ is a bijective function. Next, for any
$a_1,a_2\in A$, $x\in X$ and $f\in FD(\varphi)$ we have that
\[
\begin{aligned}
(E_{a_1},E_{a_2})\in \delta_x^{A/E}\ &\iff\ (a_1,a_2)\in E\circ \delta_x^A\circ E\ \iff\ (a_1,a_2)\in \varphi \circ \delta_x^B\circ \varphi^{-1} \\
&\iff\ (\exists b_1,b_2\in B)\,\bigl( (a_1,b_1)\in\varphi \land (b_1,b_2)\in \delta_x^B \land (a_2,b_2)\in \varphi \bigr) \\
&\iff\ (\exists b_1,b_2\in B)\,\bigl( (f(a_1),b_1)\in F \land (b_1,b_2)\in \delta_x^B \land (f(a_2),b_2)\in F \bigr) \\
&\iff\ (f(a_1),f(a_2))\in F\circ\delta_x^B\circ F \ \iff\ (F_{f(a_1)},F_{f(a_2)})\in \delta_x^{B/F} \\
&\iff\ (\widetilde\varphi(E_{a_1}),\widetilde\varphi(E_{a_2}))\in \delta_x^{B/F}.
\end{aligned}
\]
and for any $a\in A$ and $f\in FD(\varphi)$ we have
\[
\begin{aligned}
E_a\in \sigma^{A/E}\ &\iff\ a\in \sigma^A\circ E \ \iff\
(\exists a'\in A)\,\bigl(a'\in \sigma^A \land (a',a)\in E\bigr) \\
&\iff\ (\exists a'\in A)\,\bigl(a'\in \sigma^A \land (a',f(a))\in \varphi\bigr)\ \iff\ f(a)\in
\sigma^A\circ \varphi= \sigma^B\circ\varphi^{-1}\circ\varphi =\sigma^B\circ F \\
&\iff\ F_{f(a)}\in \sigma^{B/F} \ \iff \ \widetilde\varphi (E_a)\in \sigma^{B/F} , \\
E_a\in \tau^{A/E}\ &\iff\ a\in E\circ\tau ^A\ \iff\ (\exists a'\in A)\,\bigl((a,a')\in E \land a'\in \tau^A\bigr)\\
&\iff\ (\exists a'\in A)\,\bigl((f(a),a')\in \varphi^{-1} \land a'\in \tau^A\bigr)\ \iff\ f(a)\in \varphi^{-1}\circ\tau^A=\varphi^{-1}\circ \varphi\circ\tau^B =F\circ \tau^B \\
&\iff\ F_{f(a)}\in \tau^{B/F} \ \iff \ \widetilde\varphi (E_a)\in \tau^{B/F} . \\
\end{aligned}
\]
Therefore, $\widetilde\varphi $ is an isomorphism between automata
${\cal A}/E$ and ${\cal B}/F$.

Conversely, let (i), (ii) and (iii) hold.~According to (i), for each $x\in X$ we have
\[
E\circ\delta_x^A\circ E=\delta_x^A\circ E= \delta_x^A\circ\varphi\circ\varphi^{-1},
\]
and by (iii), for arbitrary $a_1,a_2\in A$ and $f\in FD(\varphi )$ we~obtain that
\[
\begin{aligned}
(a_1,a_2)\in \delta_x^A\circ\varphi\circ\varphi^{-1}\ &\iff\ (a_1,a_2)\in E\circ \delta_x^A\circ E
\ \iff\ (E_{a_1},E_{a_2})\in \delta_x^{A/E} \\
&\iff\ (\widetilde\varphi(E_{a_1}),\widetilde\varphi(E_{a_2}))\in \delta_x^{B/F}\ \iff\
(F_{f(a_1)},F_{f(a_2)})\in \delta_x^{B/F}\\
&\iff\
(f(a_1),f(a_2))\in F\circ \delta_x^B\circ F \\
&\iff\ (\exists b_1,b_2\in B)\, \bigl( (f(a_1),b_1)\in F  \land (b_1,b_2)\in \delta_x^B \land
(f(a_2),b)\in F\bigr) \\
&\iff\ (\exists b_1,b_2\in B)\, \bigl( (a_1,b_1)\in\varphi \land (b_1,b_2)\in \delta_x^B \land
(a_2,b)\in \varphi\bigr) \\
&\iff\ (a_1,a_2)\in \varphi\circ\delta_x^B\circ\varphi^{-1}.
\end{aligned}
\]
Therefore, the first equality in (\ref{eq:ufb0.1}) holds. In a similar way we prove the second equality in (\ref{eq:ufb0.1}).

Next, for every $a\in A$ we have that
\[
\begin{aligned}
a\in \sigma^A\circ\varphi\circ\varphi^{-1}\ &\iff\ a\in \sigma^A\circ E\ \iff\ E_a\in \sigma^{A/E}\ \iff\
\widetilde\varphi (E_a)\in \sigma^{B/F}\ \iff\ F_{f(a)}\in \sigma^{B/F}\\
&\iff\ f(a)\in \sigma^B\circ F\ \iff\ (\exists b\in B)\,\bigl(b\in \sigma^B \land (f(a),b)\in F \bigr)\\
&\iff\ (\exists b\in B)\,\bigl(b\in \sigma^B \land (a,b)\in \varphi \bigr)\ \iff\ a\in \sigma^B\circ \varphi^{-1},
\end{aligned}
\]
so $\sigma^A\circ\varphi\circ\varphi^{-1}=\sigma^B\circ \varphi^{-1}$, and hence,
$\sigma^A\circ\varphi=\sigma^B\circ \varphi^{-1}\circ\varphi$. For every $a\in A$ we also have
\[
\begin{aligned}
a\in \tau^A\ &\iff\ a\in E\circ\tau^A\ \iff\ E_a\in \tau^{A/E}\ \iff\ \widetilde\varphi(E_a)\in\tau^{B/F}
\ \iff\ F_{f(a)}\in \tau^{B/F}\ \iff\ f(a)\in F\circ\tau^B\\
&\iff\ (\exists b\in B)\,\bigl((f(a),b)\in F\land b\in \tau^B\bigr)\ \iff\
(\exists b\in B)\,\bigl((a,b)\in \varphi\land b\in \tau^B\bigr)\ \iff\ a\in \varphi\circ\tau^B,
\end{aligned}
\]
whence $\tau^A=\varphi\circ\tau^B$. Likewise, $\tau^B=\varphi^{-1}\circ\tau^A$. Therefore, we have proved
that (\ref{eq:ufb0.3}) and (\ref{eq:ufb0.4}) also hold, and consequently, $\varphi $ is a forward bisimulation.
\end{proof}

The question that naturally arises is under what conditions two given forward~bisimulation equivalences on two automata determine a uniform forward bisimulation.~An answer to this question is given by the~fol\-lowing theorem.

\begin{theorem}\label{th:ufb.ex}
Let ${\cal A}=(A,\delta^A,\sigma^A,\tau^A)$ and ${\cal B}=(B,\delta^B,\sigma^B,\tau^B)$ be
automata,~and let $E$ and $F$ be forward~bisimulation equivalences on $\cal A$ and $\cal B$.

Then there exists a uniform forward bisimulation $\varphi\subseteq A\times
B$~such that $E_A^\varphi =E$ and $E_B^\varphi =F$ if and only if the factor~automata
${\cal A}/E$ and ${\cal B}/F$ are isomor\-phic.
\end{theorem}

\begin{proof}
The direct part of the theorem is an immediate consequence of Theorem \ref{th:ufb}.

Conversely, let $\phi:A/E\to B/F$ be an isomorphism between factor automata ${\cal A}/E$ and
${\cal B}/F$. Let us define $\varphi \subseteq A\times B$ as in (\ref{eq:varphi.phi}), i.e.,
\[
(a,b)\in \varphi \ \iff\ \phi(E_a)=F_b, \ \ \text{for all $a\in A$ and $b\in B$}.
\]
By the proof of Theorem \ref{th:ur2}, $\varphi $ is a uniform relation such that $E=E_A^\varphi $,
$F=E_B^\varphi $ and $\phi=\widetilde\varphi $, and according to Theorem \ref{th:ufb}, $\varphi $
is a forward bisimulation.
\end{proof}

Next we prove the following.

\begin{theorem}\label{th:G:E.2}
Let ${\cal A}=(A,\delta^A,\sigma^A,\tau^A)$ be an automaton, let $E$ be a forward
bisimulation equivalence on $\cal A$, and let $F$ be an equivalence on $A$ such that $E\subseteq F$.

Then $F$ is a forward bisimulation equivalence on $\cal A$ if and only if $F/E$ is a forward bisimulation
equivalence~on~${\cal A}/E$.
\end{theorem}

\begin{proof}
As in the proof of Theorem \ref{th:F:E}, set $F/E=P$.~For arbitrary $a_1,a_2\in A$ and $x\in  X$, by~the proof of Theorem \ref{th:F:E}
we obtain that
\[
(E_{a_1},E_{a_2})\in P\circ \delta_x^{A/E}\circ P \ \iff\ (a_1,a_2)\in F\circ \delta_x^A\circ F,
\]
and also,
\[
\begin{aligned}
(E_{a_1},E_{a_2})\in \delta_x^{A/E}\circ P\ &\iff\ (\exists a_3\in A)\ \bigl((E_{a_1},E_{a_3})\in \delta_x^{A/E}
\land (E_{a_3},E_{a_2})\in P \bigr)\\
&\iff\ (\exists a_3\in A)\ \bigl((a_1,a_3)\in E\circ \delta_x^A\circ E
\land (a_3,a_2)\in F \bigr)\\
&\iff\ (a_1,a_2)\in E\circ \delta_x^A\circ E\circ F \\
&\iff\ (a_1,a_2)\in \delta_x^A\circ F ,
\end{aligned}
\]
since $E\circ \delta_x^A\circ E\circ F=\delta_x^A\circ E\circ F=\delta_x^A\circ F$. Therefore,
\[
P\circ \delta_x^{A/E}\circ P = \delta_x^{A/E}\circ P \ \iff\ F\circ \delta_x^A\circ F = \delta_x^A\circ F .
\]
Furthermore, for an arbitrary $a\in A$ we have that
\[
\begin{aligned}
E_{a}\in P\circ \tau^{A/E}\ &\iff\ (\exists a'\in A)\ (E_a,E_{a'})\in P \land E_{a'}\in \tau^{A/E}&\iff\ (\exists a'\in A)\ (a,{a'})\in F \land {a'}\in E\circ\tau^{A}\\
&\iff\ a\in F\circ E\circ \tau^A=F\circ\tau^{A},
\end{aligned}
\]
and according to (\ref{eq:tauAE}) and (\ref{eq:fb.ont}), $E_a\in \tau^{A/E}\ \iff\ a\in E\circ \tau^A=\tau^A$.
Hence,
\[
P\circ \tau^{A/E}=\tau^{A/E}\ \iff\ F\circ \tau^A=\tau^A,
\]
proving our claim.
\end{proof}

In view of Theorem \ref{th:F:E}, the rule $F\mapsto F/E$ defines an isomorphism between lattices
${\cal E}_E(A)$ and ${\cal E}(A/E)$, for every $E\in {\cal E}(A)$.~According to Theorem \ref{th:G:E.2}, the same
rule determines an isomorphism~between lattices ${\cal E}_E^{\mathrm{fb}}({\cal A})$~and
${\cal E}^{\mathrm{fb}}({\cal A}/E)$,~where
${\cal E}_E^{\mathrm{fb}}({\cal A})=\{F\in {\cal E}^{\mathrm{fb}}({\cal A})\mid E\subseteq F\}$, for each
$E\in {\cal E}^{\mathrm{fb}}({\cal A})$.~

Consequently, the following is~true.

\begin{corollary}\label{cor:F:E.g}
Let ${\cal A}=(A,\delta^A,\sigma^A,\tau^A)$ be an automaton, and let $E$ and $F$ be
forward bisimulation equivalences on $\cal A$ such that $E\subseteq F$.

Then $F$ is the greatest forward bisimulation equivalence on $\cal A$ if and only if $F/E$ is the greatest
forward bisimulation equivalence on ${\cal A}/E$.
\end{corollary}

\begin{proof}
This follows immediately by Theorems \ref{th:G:E.2} and equation (\ref{eq:ord.isom}).
\end{proof}

\section{Forward bisimulation equivalent automata}

Let ${\cal A}=(A,\delta^A,\sigma^A,\tau^A)$ and ${\cal B}=(B,\delta^B,\sigma^B,\tau^B)$ be
automata.~If there exists~a complete and surjective forward bisimula\-tion from
$\cal A$ to $\cal B$, then we say that $\cal A$ and $\cal B$ are {\it forward bisimulation~equi\-valent\/},
or briefly {\it FB-equivalent\/}, and we write ${\cal A}\sim_{FB}{\cal B}$.~Notice that completeness
and surjectivity of this forward bisimulation mean that every state of $\cal A$ is equivalent to some state
of $\cal B$, and vice versa.~For
any automata $\cal A$, $\cal B$
and $\cal C$ we have that
\begin{equation}\label{eq:fbe.eq}
{\cal A}\sim_{FB}{\cal A};\ \ \ \ {\cal A}\sim_{FB}{\cal B} \implies {\cal B}\sim_{FB}{\cal A};\ \ \ \
\bigl({\cal A}\sim_{FB}{\cal B} \land {\cal B}\sim_{FB}{\cal C}\bigr) \implies {\cal A}\sim_{FB}{\cal C}.
\end{equation}
Similarly, we call $\cal A$ and $\cal B$ {\it backward bisimulation equivalent\/}, briefly {\it BB-equiva\-lent\/}, in notation
${\cal A}\sim_{BB}{\cal B}$, if there exists a complete and surjective backward bisimulation from
$\cal A$ to $\cal B$.

First we prove that every automaton $\cal A$ is FB-equivalent to the factor automaton of $\cal A$ with respect to any forward bisimulation equivalence on $\cal A$.

\begin{theorem}\label{th:nat.fb}
Let ${\cal A}=(A,\delta^A,\sigma^A,\tau^A)$ be an automaton, let $E$ be an equivalence
on~$A$, let $\varphi_E$ be the natural function from $A$ to $A/E$, and let ${\cal A}/E=(A/E,\delta^{A/E},\sigma^{A/E},\tau^{A/E})$ be the factor automaton
of $\cal A$ with respect to~$E$.

Then $\varphi_E$ is both a forward  and a backward simulation.

Moreover, the following conditions are equivalent:
\begin{itemize}\parskip-2pt\itemindent12pt
\item[{\rm (i)}] $E$ is a forward bisimulation on $\cal A$;
\item[{\rm (ii)}] $\varphi_E$ is a forward bisimulation;
\item[{\rm (iii)}] $\varphi_E$ is a backward-forward bisimulation.
\end{itemize}
\end{theorem}

\begin{proof} Note that for arbitrary $a_1,a_2\in A$ we have that $\varphi_E(a_1)=E_{a_2}$ (i.e., $(a_1,E_{a_2})\in \varphi_E$) if and only if $(a_1,a_2)\in E$.

For arbitrary $x\in X$ and $a_1,a_2\in A$ we have that
\begin{align}
(a_1,E_{a_2})\in \delta_x^A\circ \varphi_E\ &\iff\ (\exists a_3\in A)\ \bigl((a_1,a_3)\in \delta_x^A \land
(a_3,E_{a_2})\in \varphi_E \bigr) \notag \\
&\iff\ (\exists a_3\in A)\ \bigl((a_1,a_3)\in \delta_x^A \land (a_3,a_2)\in E \bigr) \notag \\
&\iff\ (a_1,a_2)\in \delta_x^A\circ E  \notag \\
&\implies\ (a_1,a_2)\in E\circ \delta_x^A\circ E=E\circ E\circ \delta_x^A\circ E \label{eq:implic}\\
&\iff\ (\exists a_3\in A)\ \bigl((a_1,a_3)\in E \land (a_3,a_2)\in (E\circ \delta_x^A\circ E)\bigr) \notag \\
&\iff\ (\exists a_3\in A)\ \bigl((a_1,E_{a_3})\in \varphi_E\land (E_{a_3},E_{a_2})\in \delta_x^{A/E}
\bigr) \notag \\
&\iff\ (a_1,E_{a_2})\in \varphi_E\circ \delta_x^{A/E} , \notag
\end{align}
and hence, $\delta_x^A\circ \varphi_E\subseteq \varphi_E\circ\delta_x^{A/E}$.~In a similar way we prove that
$\varphi_E^{-1}\circ \delta_x^A\subseteq \delta_x^{A/E}\circ \varphi_E^{-1}$.

Furthermore, for any $a\in A$ we have that
\[
a\in \sigma^A\ \implies\ E_a\in \sigma^{A/E} \land (E_a,a)\in \varphi_E^{-1} \ \implies\ a\in \sigma^{A/E}\circ \varphi_E^{-1},
\]
whence $\sigma^{A}\subseteq \sigma^{A/E}\circ \varphi_E^{-1}$, and
\[
\begin{aligned}
E_a\in \sigma^A\circ \varphi_E\ &\iff\ (\exists a'\in A)\ a'\in \sigma^A\land (a',E_a)\in \varphi_E\ \iff\
(\exists a'\in A)\ a'\in \sigma^A\land (a',a)\in E\\
&\iff\  a\in \sigma^A\circ E\ \iff\ E_a\in \sigma^{A/E},
\end{aligned}
\]
what yields $\sigma^A\circ \varphi_E\subseteq \sigma^{A/E}$. In a similar way we show that $\varphi_E^{-1}\circ \tau^A\subseteq \tau^{A/E}$ and $\tau^A\subseteq \varphi_E\circ \tau^{A/E}$.

Therefore, we have proved that $\varphi_E $ is both a forward and a backward simulation.

Moreover, we have that the opposite implication in (\ref{eq:implic}) holds (i.e.,
$\varphi_E^{-1}$ is a forward simulation) if and only if $E$ is a forward bisimulation on $\cal A$.~This proves the equivalence of the conditions (i), (ii), and (iii).
\end{proof}

Now we state and prove the main result of this section.

\begin{theorem}\label{th:UFBeq}
Let ${\cal A}=(A,\delta^A,\sigma^A,\tau^A)$ and
${\cal B}=(B,\delta^B,\sigma^B,\tau^B)$ be automata, and let $E$ and $F$ be
the greatest forward bisimulation equivalences on $\cal A$ and $\cal B$.

Then $\cal A$ and $\cal B$ are FB-equivalent if and only if factor automata ${\cal A}/E$ and
${\cal B}/F$ are isomorphic.
\end{theorem}

\begin{proof}
Let $\cal A$ and $\cal B$ be FB-equivalent automata, i.e., let there exists a complete
and~surjective forward bisimulation $\psi \subseteq A\times B$.~According to Theorem \ref{th:gfb}, then there
exists the greatest~forward bisimulation $\varphi $ from $\cal A$ to $\cal B$, and $\varphi $ is a
partial uniform relation.~Since $\psi $ is complete and surjective, and $\psi\subseteq \varphi $, then
$\varphi $ is also complete and surjective, what means that $\varphi $ is a uniform forward bisimulation.

By Theorem \ref{th:ufb}, $E_A^\varphi $ and
$E_B^\varphi $ are forward bisimulation equivalences on $\cal A$ and~$\cal B$, and $\widetilde\varphi $ is
an isomorphism of factor automata ${\cal A}/E_A^\varphi $ and ${\cal B}/E_B^\varphi $.~Let $P$ and $Q$
denote respectively the greatest forward bisimulation equivalences on ${\cal A}/E_A^\varphi $ and
${\cal B}/E_B^\varphi $. By the fact that $\widetilde\varphi $ is an isomorphism of ${\cal A}/E_A^\varphi $
onto ${\cal B}/E_B^\varphi $ we obtain that $P$ and $Q$ are related by
\[
(\alpha_1,\alpha_2)\in P \ \iff\ \bigl(\widetilde\varphi(\alpha_1),\widetilde\varphi(\alpha_2)\bigr)\in Q,
\ \ \ \text{for all $\alpha_1,\alpha_2\in A/E_A^\varphi $},
\]
so we can define an isomorphism $\xi:({\cal A}/E_A^\varphi)/P\to ({\cal B}/E_B^\varphi )/Q$ by
$\xi (P_\alpha)=Q_{\widetilde\varphi (\alpha)}$, for every $\alpha\in A/E_A^\varphi $.

Now, according to Corollary \ref{cor:F:E.g}, $P=E/E_A^\varphi $ and $Q=F/E_B^\varphi $, and by Theorem
\ref{th:F:E} we obtain
\[
{\cal A}/E\cong ({\cal A}/E_A^\varphi)/P\cong ({\cal B}/E_B^\varphi )/Q\cong {\cal B}/F,
\]
what was to be proved.

The converse follows immediately by Theorem \ref{th:ufb.ex}.
\end{proof}

As a direct consequence of previous two theorems we obtain the following.

\begin{corollary}\label{cor:class.min}
Let $\cal A$ be an automaton, let $E$ be the greatest forward bisimulation
equivalence on $\cal A$, and let $\Bbb{FB}(A)$ be the class of all automata which are FB-equivalent~to~$\cal A$.

Then ${\cal A}/E$ is the unique {\rm({\it up to an isomorphism\/})} minimal automaton in $\Bbb{FB}(A)$.
\end{corollary}

\begin{proof}
Let $\cal B$ be any minimal automaton from $\Bbb{FB}(A)$, and let $F$ be the greatest forward
bisimulation~equivalence on $\cal B$.~According to Theorem \ref{th:nat.fb} and (\ref{eq:fbe.eq}),
${\cal B}/F$ also belongs to $\Bbb{FB}(A)$, and by minimality of $\cal B$ it follows that $F$ is the
equality relation.~Now, by Theorem \ref{th:UFBeq} we obtain that ${\cal B}\cong {\cal B}/F\cong {\cal A}/E$, proving
our claim.
\end{proof}

According Theorem \ref{th:UFBeq}, the problem of testing FB-equivalence of two automata $\cal A$~and~$\cal B$ can be~reduced to the problem of testing isomorphism of their factor automata with respect to the greatest forward bisimulation equivalences on $\cal A$~and~$\cal B$.~It is worth of mention that the isomorphism problem for nondeterministic automata is equivalent to the well-known {\it graph isomorphism problem\/},~the computational problem of determining whether two finite graphs~are isomorphic.~Besides its practi\-cal importance, the graph isomorphism problem is a curiosity in computational complexity theory, as it is one of a very small number of problems belonging~to NP that is neither known to be computable in polynomial time nor NP-complete.~Along~with integer fac\-torization, it is one of the few important algorith\-mic problems whose rough computational~com\-plexity~is still not known, and it is generally accepted that graph isomorphism is a problem that lies~between P and NP-complete if P$\ne $NP (cf.~\cite{Skiena.08}).~However, although no worst-case polynomial-time algorithm is known, testing graph isomorphism is usually not very hard in practice.~The basic algorithm examines all $n!$~possible~bijec\-tions between the nodes of two graphs (with $n$ nodes), and tests whether they preserve adjacency of the nodes.~Clearly, the major problem is the rapid growth in the number of bijections when the number of nodes is growing,~which is also the crucial problem in testing isomorphism between fuzzy automata, but the algorithm can be~made more efficient by suitable partitioning of the sets of nodes as described in \cite{Skiena.08}.~What is good in our case is that the isomorphism test is applied not to the automata $\cal A$ and $\cal B$, but to the factor automata with respect to the greatest forward bisimulation equivalences on ${\cal A}$ and ${\cal B}$. The number of states of these factor automata can be much smaller than the number of states of $\cal A$ and~$\cal B$, which can significantly affect the duration of testing.

According to Lemma \ref{le:lang.incl.eq}, FB-equivalent automata are language equivalent, but the
converse does not hold, as the following example shows.

\begin{example}\rm
Let ${\cal A}=(A,\delta^A,\sigma^A,\tau^A)$ and ${\cal B}=(B,\delta^B,\sigma^B,\tau^B)$ be
automata with $|A|=3$, $|B|=2$ and $X=\{x\}$, whose transition relations and sets of
initial and terminal states are represented by the following Boolean matrices and vectors:
\[
\delta_x^A=\begin{bmatrix}
1 & 0 & 0 \\
0 & 0 & 1 \\
0 & 0 & 0
\end{bmatrix},\ \
\sigma^A=\begin{bmatrix}
0 & 1 & 0
\end{bmatrix},\ \
\tau^A=\begin{bmatrix}
0 \\
0 \\
1
\end{bmatrix},\ \ \ \
\delta_x^B=\begin{bmatrix}
0 & 1 \\
0 & 0
\end{bmatrix},\ \
\sigma^B=\begin{bmatrix}
1 & 0
\end{bmatrix},\ \
\tau^B=\begin{bmatrix}
0 \\
1
\end{bmatrix}.
\]
These automata are language-equivalent, both of them recognize the language $L=\{x\}$.~On
the other hand,~the greatest forward bisimulation equivalences $E$ on $\cal A$ and $F$ on $\cal B$ are
equality relations, so ${\cal A}/E\cong {\cal A}$ and ${\cal B}/F\cong {\cal B}$.~But, ${\cal A}$ and
${\cal B}$ have different number of states, and hence, they are~not~isomorphic.~Therefore, according to
Theorem \ref{th:UFBeq}, $\cal A$ and $\cal B$ are not FB-equivalent.

\end{example}

\section{Uniform backward-forward bisimulations}

In this section we consider uniform backward-forward bisimulations.~We will see that they have certain properties similar to the corresponding properties of uniform forward bisimulations, but we will also show that there are some essential differences.

First we prove the following analogue of Theorem \ref{th:ufb}.

\begin{theorem}\label{th:ubfb}
Let ${\cal A}=(A,\delta^A,\sigma^A,\tau^A)$ and ${\cal B}=(B,\delta^B,\sigma^B,\tau^B)$ be automata and let $\varphi \subseteq A\times B$ be a uniform relation. Then $\varphi $ is a backward-forward bisimulation if and only if the following hold:
\begin{itemize}\parskip=0pt\itemindent6pt
\item[{\rm (i)}] $E_A^\varphi $ is a forward bisimulation equivalence on $\cal A$;
\item[{\rm (ii)}] $E_B^\varphi $ is a backward bisimulation equivalence on $\cal B$;
\item[{\rm (iii)}]  $\widetilde \varphi $ is an isomorphism of factor automata ${\cal A}/E_A^\varphi $ and ${\cal B}/E_B^\varphi $.
\end{itemize}
\end{theorem}

\begin{proof}
For the sake of simplicity set $E=E_A^\varphi $ and $F=E_B^\varphi $.~According to Theorem \ref{th:ur}, we have that $E=\varphi\circ\varphi^{-1}$ and $F=\varphi^{-1}\circ\varphi$.

Let $\varphi $ be a backward-forward bisimulation.~Then
\[
\begin{aligned}
&E\circ \delta_x^A\circ E = \varphi\circ\varphi^{-1}\circ\delta_x^A\circ \varphi\circ\varphi^{-1} = \varphi\circ\varphi^{-1}\circ\varphi\circ \delta_x^B\circ \varphi^{-1} = \varphi\circ \delta_x^B\circ \varphi^{-1} =
\delta_x^A\circ \varphi\circ\varphi^{-1} = \delta_x^A\circ E, \\
&E\circ \tau^A =  \varphi\circ\varphi^{-1}\circ \tau^A = \varphi\circ\varphi^{-1}\circ \varphi\circ\tau^B =
\varphi\circ\tau^B = \tau^A, \\
&F\circ \delta_x^B\circ F = \varphi^{-1}\circ\varphi\circ \delta_x^B\circ \varphi^{-1}\circ\varphi =
\varphi^{-1}\circ \delta_x^A\circ\varphi\circ \varphi^{-1}\circ\varphi = \varphi^{-1}\circ \delta_x^A\circ\varphi =
\varphi^{-1}\circ\varphi \circ \delta_x^B = F\circ \delta_x^B , \\
&\sigma^B\circ F = \sigma^B\circ \varphi^{-1}\circ\varphi = \sigma^A\circ\varphi \circ \varphi^{-1}\circ\varphi =
\sigma^A\circ\varphi = \sigma^B .
\end{aligned}
\]
Hence, $E=E_A^\varphi$ is a forward bisimulation equivalence on $\cal A$ and $F=E_B^\varphi$ is a backward bisimulation~equivalence on $\cal B$.~As in the proof of Theorem \ref{th:ufb} we show that $\widetilde{\varphi}$ is an isomorphism of automata  ${\cal A}/E$ and~${\cal B}/F$.

Conversely, let (i), (ii), and (iii) hold.~For every $\psi \in FD(\varphi )$, $\xi \in FD(\varphi^{-1})$, $a_1,a_2\in A$, $b_1,b_2\in B$ and $x\in X$, as in the proof of Theorem \ref{th:ufb} we show that
\[
\begin{aligned}
&(a_1,a_2)\in (E\circ \delta_x^A\circ E)\Leftrightarrow (\psi(a_1),\psi(a_2))\in (F\circ \delta_x^B\circ F) ,\\
&(b_1,b_2)\in (F\circ \delta_x^B\circ F)\Leftrightarrow (\xi(b_1),\xi(b_2))\in (E\circ \delta_x^A\circ E) ,
\end{aligned}
\]
and by (i) and (ii) we obtain that
\[
\begin{aligned}
&\delta_x^A\circ\varphi =\delta_x^A\circ\varphi\circ \varphi^{-1}\circ\varphi = \delta_x^A\circ E\circ\varphi =
E\circ \delta_x^A\circ E\circ\varphi = E\circ \delta_x^A\circ\varphi , \\
&\varphi \circ\delta_x^B= \varphi \circ \varphi^{-1}\circ\varphi \circ\delta_x^B = \varphi \circ F \circ\delta_x^B =
\varphi \circ F \circ\delta_x^B\circ F = \varphi  \circ\delta_x^B \circ F .
\end{aligned}
\]
Now, for all $a\in A$ and $b\in B$ we obtain that
\[
\begin{aligned}
(a,b)\in \delta_x^A\circ \varphi &\Leftrightarrow (a,b)\in E\circ \delta_x^A\circ\varphi  \Leftrightarrow (\exists a_1\in A)\ ((a,a_1)\in E\circ \delta_x^A\land  (a_1,b)\in\varphi) \\
 &\Leftrightarrow(\exists a_1\in A)\ ((a,a_1)\in E\circ \delta_x^A\land (a_1,\xi (b))\in E) \Leftrightarrow(a,\xi(b)) \in E\circ \delta_x^A\circ E\\
&\Leftrightarrow (\psi (a),\psi(\xi(b)))\in F\circ \delta_x^B\circ F \Leftrightarrow(\psi (a),b) \in F\circ \delta_x^B\circ F \\
&\Leftrightarrow (\exists b_1\in B)\ ((\psi(a),b_1)\in F\land (b_1,b)\in\delta_x^B\circ F) \Leftrightarrow (\exists b_1\in B)\ ((a,b_1)\in\varphi\land (b_1,b)\in\delta_x^B\circ F)\\
& \Leftrightarrow (a,b)\in\varphi\circ \delta_x^B\circ F \Leftrightarrow (a,b)\in\varphi\circ \delta_x^B,
\end{aligned}
\]
and hence, $\delta_x^A\circ\varphi = \varphi\circ \delta_x^B$.~As in the proof of Theorem \ref{th:ufb} we prove that $\tau^A=\varphi\circ \tau^B$, and analogously we obtain that $\sigma^A\circ\varphi =\sigma^B$.~Therefore, $\varphi $ is a forward-backward bisimulation.
\end{proof}

We can also prove the following analogue of Theorem \ref{th:ufb.ex}.

\begin{theorem}\label{th:ubfb.ex}
Let ${\cal A}=(A,\delta^A,\sigma^A,\tau^A)$ and ${\cal B}=(B,\delta^B,\sigma^B,\tau^B)$ be automata, let $E$ be a forward bisimulation equivalence on $\cal A$ and $F$ a backward bisimulation equivalence on $\cal B$.

Then there exists a uniform backward-forward bisimulation $\varphi\subseteq A\times B$~such that $E_A^\varphi =E$ and
$E_B^\varphi =F$ if and only if factor automata ${\cal A}/E$ and ${\cal B}/F$ are isomor\-phic.
\end{theorem}

\begin{proof}
This theorem can be proved in a similar way as Theorem \ref{th:ufb.ex}.
\end{proof}

In Theorem \ref{th:nat.fb} we proved that for any equivalence $E$, its natural function $\varphi_E$ is a forward bisimulaton if and only if it is a backward-forward bisimulation.~Now we prove a more general theorem, which shows that this holds for an arbitrary function.

\begin{theorem}\label{th:func}
Let ${\cal A}=(A,\delta^A,\sigma^A,\tau^A)$ and ${\cal B}=(B,\delta^B,\sigma^B,\tau^B)$ be automata, let $\varphi : A\to B$ be a function, and let $E=E_A^\varphi $ be the kernel of $\varphi $.~Then the following conditions are equivalent:
\begin{itemize}\parskip=0pt
\item[{\rm (i)}] $\varphi $ is a forward bisimulation;
\item[{\rm (ii)}] $\varphi $ is a backward-forward bisimulation;
\item[{\rm (iii)}] $E$ is a forward bisimulation equivalence on $\cal A$ and the function $\phi :A/E\to B$ given by $\phi (E_a)=\varphi (a)$, for each $a\in A$, is a monomorphism of the factor automaton ${\cal A}/E$ into $\cal B$.
\end{itemize}
\end{theorem}

\begin{proof}
Let $C=\im\varphi $ and consider the subautomaton ${\cal C}=(C,\delta^C,\sigma^C,\tau^C)$ of $\cal B$.

(i)$\implies $(iii). According to Lemma \ref{le:restr}, $\varphi \subseteq A\times C$ and $\varphi $ is a forward bisimulation from $\cal A$ to $\cal C$.~We also have that $\varphi $ is a surjective function from $A$ onto $C$, and hence, it is a uniform relation from $A$ to $C$.~Now, by Theorem \ref{th:ufb} we obtain that $E=E_A^\varphi $ is a forward bisimulation equivalence on $\cal A$, $E_C^\varphi $ is the equality relation on $C$, and $\widetilde \varphi $ is an isomorphism from ${\cal A}/E $ to ${\cal C}/E_B^\varphi \cong {\cal C}$.~If we identify ${\cal C}/E_B^\varphi $ and ${\cal C}$, then it is easy to see that $\widetilde \varphi $ can be represented as $\phi $, where $\phi $ is defined as in (iii), so $\phi $ is a monomorphism of ${\cal A}/E$ into~$\cal B$.

(iii)$\implies $(i). This is a direct consequence of Theorem \ref{th:ufb}, since $E_C^\varphi $ is the equality relation and $\widetilde \varphi $ and $\phi $ can be identified.

(i)$\iff $(ii). This follows immediately by Theorems \ref{th:ufb} and \ref{th:ubfb}, since $E_C^\varphi $ is the equality relation on $C$, and it is both a forward and backward bisimulation equivalence.
\end{proof}

\section{Weak simulations and bisimulations}

In this section we introduce and study two new types of bisimulations, which are more general than forward and backward bisimulations.

Let ${\cal A}=(A,X,\delta^A,\sigma^A,\tau^A)$ be an automaton.~For each $u\in X^*$ we define subsets $\sigma_u^A$ and $\tau_u^A$ of $A$ as follows:
\begin{equation}\label{eq:s.u-t.u}
\sigma_u^A=\sigma^A\circ \delta_u^A, \qquad \tau_u^A=\delta_u^A\circ \tau^A .
\end{equation}
Moreover, for each $a\in A$, the {\it right language\/} $\overrightarrow{L}_{\cal A}(a)$ and the {\it left language\/} $\overleftarrow{L}_{\cal A}(a)$ of the state $a$
are languages
\begin{equation}\label{eq:r.l.lang}
\overrightarrow{L}_{\cal A}(a)=\{u\in X^*\mid a\in \tau_u^A\}, \ \ \ \ \overleftarrow{L}_{\cal A}(a)=\{u\in X^*\mid a\in \sigma_u^A\}.
\end{equation}
In other words, the right language of $a$ is the language recognized by the automaton obtained from $\cal A$ by replacing $\sigma^A$ by $\{a\}$,
and the left language of $a$ is the language recognized by the automaton obtained from $\cal A$ by replacing $\tau^A$ by $\{a\}$.
When the automaton $\cal A$ is known from the context, we omit the subscript $\cal A$, and we write just $\overrightarrow{L}(a)$ and $\overleftarrow{L}(a)$.

Now, let ${\cal A}=(A,\delta^A,\sigma^A,\tau^A)$ and ${\cal B}=(B,\delta^B,\sigma^B,\tau^B)$ be automata
and let $\varphi\subseteq A\times B$ be~a non-empty~rela\-tion. We call $\varphi$ a {\it
weak forward simulation\/} from $\cal A$ to $\cal B$ if
\begin{align}
&\varphi^{-1}\circ \tau_u^A\subseteq \tau_u^B,\ \ \  \text{for every $u\in X^*$},\label{eq:wfs.tau}\\
&\sigma^A\subseteq \sigma^B\circ \varphi^{-1},\label{eq:wfs.sigma}
\end{align}
and we call $\varphi $ a {\it weak backward simulation\/} from $\cal A$ to $\cal B$ if
\begin{align}
&\sigma_u^A\circ \varphi \subseteq \sigma_u^B,\ \ \  \text{for every $u\in X^*$}, \label{eq:wbs.sigma} \\
&\tau^A\subseteq \varphi\circ\tau^B. \label{eq:wbs.tau}
\end{align}
We call $\varphi $ a {\it weak forward bisimulation\/} if both $\varphi $ and $\varphi^{-1}$ are
weak forward simulations,~that is, if it satisfies (\ref{eq:wfs.tau}), (\ref{eq:wfs.sigma}), and
\begin{align}
&\varphi\circ \tau_u^B\subseteq \tau_u^A,\ \ \ \text{for every $u\in X^*$},\label{eq:wfsi.tau}\\
&\sigma^B\subseteq \sigma^A\circ \varphi , \label{eq:wfsi.sigma}
\end{align}
and we call $\varphi $ a {\it weak backward bisimulation\/} if both $\varphi $ and $\varphi^{-1}$ are weak backward simulations, that is, if it satisfies (\ref{eq:wbs.sigma}), (\ref{eq:wbs.tau}), and
\begin{align}
&\sigma_u^B\circ \varphi^{-1}\subseteq \sigma_u^A,\ \ \ \text{for every $u\in X^*$},\label{eq:wbsi.sigma}\\
&\tau^B\subseteq \varphi^{-1}\circ \tau^A. \label{eq:wbsi.tau}
\end{align}
For the sake of simplicity, we~will call $\varphi $ just a {\it weak simulation\/} if it is either a weak forward or a weak backward simulation, and just a {\it weak bisimulation\/} if it is either a weak forward or a weak backward bisimulation.

First we prove the following two lemmas.

\begin{lemma}\label{le:ws-b.lang}
Let ${\cal A}=(A,\delta^A,\sigma^A,\tau^A)$ and ${\cal B}=(B,\delta^B,\sigma^B,\tau^B)$ be automata, and let $\varphi \subseteq A\times B$ be a relation.~Then
\begin{itemize}\parskip=0pt
\item[{\rm (a)}] If $\varphi $ is a weak simulation, then $L({\cal A})\subseteq L({\cal B})$.
\item[{\rm (b)}] If $\varphi $ is a weak bisimulation, then $L({\cal A})= L({\cal B})$.
\item[{\rm (c)}] If $\varphi $ is a forward {\rm ({\it resp.~backward\/})} simulation, then it is a weak forward {\rm ({\it resp.~backward\/})} simulation.
\end{itemize}
\end{lemma}

\begin{proof}
(a) Let $\varphi $ be a weak forward simulation. Then for every $u\in X^*$ we have that
\[
\sigma^A\circ \delta_u^A\circ \tau^A = \sigma^A\circ \tau_u^A \leqslant \sigma^B\circ \varphi^{-1}\circ \tau_u^A \leqslant
\sigma^B\circ \tau_u^B = \sigma^B\circ\delta_u^B\circ \tau^B,
\]
and by (\ref{eq:lang.rec.2}) we obtain that $L({\cal A})\subseteq L({\cal B})$. Similarly, if $\varphi $ is a
weak backward simulation,~then $L({\cal A})\subseteq L({\cal B})$.

(b) This follows immediately by (a).

(c) Let $\varphi $ be a forward simulation.~From (\ref{eq:fs.nda.s}) it follows immediately that (\ref{eq:wfs.sigma}) holds, and by (\ref{eq:fs.nda.t}) we obtain that (\ref{eq:wfs.tau}) holds for $u=\varepsilon$.~Suppose that (\ref{eq:wfs.tau}) holds for all words of length $n$, for some natural number $n$, and consider a word $u\in X^*$ of length $n+1$, i.e., $u=xv$, for some $x\in X$ and $v\in X^*$ such that $v$ has the length $n$. Then
\[
\varphi^{-1}\circ \tau_u^A = \varphi^{-1}\circ \delta_x^A\circ \tau_v^A \subseteq \delta_x^B\circ \varphi^{-1}\circ  \tau_v^A \subseteq \delta_x^B\circ \tau_v^B = \tau_u^B .
\]
Hence, by induction we obtain that (\ref{eq:wfs.tau}) holds for every $u\in X^*$.~In a similar way we prove the assertion concerning backward simulations.
\end{proof}

\begin{lemma}\label{le:wdual}
Let ${\cal A}=(A,\delta^A,\sigma^A,\tau^A)$ and ${\cal B}=(B,\delta^B,\sigma^B,\tau^B)$ be automata. A relation $\varphi\subseteq A\times B$ is a weak backward bisimulation from $\cal A$ to $\cal B$ if and only if it is a weak forward bisimulation from $\bar{\cal A}$ to $\bar {\cal B}$.
\end{lemma}

\begin{proof}
We can easily show that $\varphi $ is a weak backward simulation from $\cal A$ to $\cal B$ if and only if $\varphi^{-1}$ is a weak forward simulation from $\bar {\cal B}$ to $\bar{\cal A}$, and $\varphi^{-1}$ is a weak backward simulation from $\cal B$ to $\cal A$ if and only if
$\varphi$ is a weak forward simulation from $\bar {\cal A}$ to $\bar{\cal B}$.
\end{proof}

According to the previous lemma, for any statement on weak forward bisimulations which~is~univer\-sally~valid (valid for all nondeterministic automata) there is the corresponding universally valid statement on weak backward bisimulations.~For that reason, we will deal only with weak forward bisimulations.

It is easy to show that the following is true.

\begin{lemma}\label{le:wfs-b.comp.union}
The composition of two weak forward simulations {\rm ({\it resp.~bisimulations\/})} and the union of an arbitrary family of weak forward simulations {\rm ({\it resp.~bisimulations\/})} are also weak forward simulations {\rm ({\it resp.~bisimulations\/})}.
\end{lemma}

Now we state and prove fundamental results concerning weak forward simulations and bisimulations. The first of them is a theorem that gives a way to decide whether there is a weak forward simulation~between two automata, and whenever it exists, provides a way to construct the greatest one.

\begin{theorem}\label{th:gwfs}
Let ${\cal A}=(A,\delta^A,\sigma^A,\tau^A)$ and ${\cal B}=(B,\delta^B,\sigma^B,\tau^B)$ be automata and define a relation $\lambda \subseteq A\times B$ by
\begin{equation}\label{eq:gwfs}
(a,b)\in \lambda \ \Leftrightarrow\ (\forall u\in X^*)\,(a\in \tau_u^A\, \Rightarrow \, b\in \tau_u^B),
\end{equation}
for all $a\in A$ and $b\in B$.

If $\lambda $ satisfies $(\ref{eq:wfs.sigma})$, then it is the greatest weak forward simulation from $\cal A$ to $\cal B$. Otherwise, if $\lambda $ does not satisfy $(\ref{eq:wfs.sigma})$, then there is no any weak forward simulation from $\cal A$ to $\cal B$.
\end{theorem}

\begin{proof} Let $\lambda $ satisfy (\ref{eq:wfs.sigma}).~If $b\in \lambda^{-1}\circ \tau_u^A$, then there exists $a\in \tau_u^A$ such that $(b,a)\in \lambda^{-1}$, and by (\ref{eq:gwfs}) we obtain that $b\in \tau_u^B$.~Therefore, $\lambda^{-1}\circ \tau_u^A\subseteq \tau_u^B$, and since $\lambda $ satisfies (\ref{eq:wfs.sigma}), we conclude that $\lambda $ is a weak forward simulation from $\cal A$ to $\cal B$.

Let $\varphi $ be an arbitrary weak forward simulation from $\cal A$ to $\cal B$, and let $(a,b)\in \varphi $.~For an arbitrary $u\in X^*$, if $a\in \tau_u^A$ then $b\in \varphi^{-1}\circ \tau_u^A\subseteq \tau_u^B$.~Therefore, we have proved that $(a,b)\in \lambda $, which means that every weak forward simulation from $\cal A$ to $\cal B$ is contained in $\lambda $.~Therefore, $\lambda $ is the greatest weak forward simulation from $\cal A$ to $\cal B$.

Suppose that $\lambda $ does not satisfy (\ref{eq:wfs.sigma}).~If $\varphi $ is an arbitrary weak forward simulation from $\cal A$ to $\cal B$, then $\sigma^A\subseteq \sigma^B\circ \varphi^{-1}\subseteq \sigma^B\circ \lambda^{-1}$, what is in contradiction with the assumption that $\lambda $ does not satisfy (\ref{eq:wfs.sigma}).~Therefore, we conclude that there is no any weak forward simulation from $\cal A$ to $\cal B$.
\end{proof}

The greatest weak forward simulation can also be represented in the following way.

\begin{corollary}\label{cor:gwfs}
Let ${\cal A}=(A,\delta^A,\sigma^A,\tau^A)$ and ${\cal B}=(B,\delta^B,\sigma^B,\tau^B)$ be automata such that there exists at least one weak forward simulation from $\cal A$ to $\cal B$, and let $\lambda $ be the greatest weak forward simulation from $\cal A$ to $\cal B$.~Then
\begin{equation}\label{eq:gwfsl}
(a,b)\in \lambda \ \Leftrightarrow\ \overrightarrow{L}(a)\subseteq \overrightarrow{L}(b),
\end{equation}
for all $a\in A$ and $b\in B$.
\end{corollary}

\begin{proof}
This is an immediate consequence of (\ref{eq:gwfs}) and the fact that $u\in \overrightarrow{L}(a)$ if and only if $a\in \tau_u^A$.
\end{proof}

The next theorem gives a way to decide whether there is a weak forward bisimu\-lation between two~auto\-mata, and if it exists, provides a way to construct the greatest one.

\begin{theorem}\label{th:gwfb}
Let ${\cal A}=(A,\delta^A,\sigma^A,\tau^A)$ and ${\cal B}=(B,\delta^B,\sigma^B,\tau^B)$ be automata and define a relation $\mu\subseteq A\times B$ by
\begin{equation}\label{eq:gwfb}
(a,b)\in \mu \ \Leftrightarrow\ (\forall u\in X^*)\,(a\in \tau_u^A\, \Leftrightarrow \, b\in \tau_u^B),
\end{equation}
for all $a\in A$ and $b\in B$.

If $\mu $ satisfies (\ref{eq:wfs.sigma}) and (\ref{eq:wfsi.sigma}), then it is the greatest weak forward bisimulation from $\cal A$ to $\cal B$, and it is a partial uniform relation. Otherwise, if $\mu $ does not satisfy (\ref{eq:wfs.sigma}) and (\ref{eq:wfsi.sigma}), then there is no any weak forward bisimulation from $\cal A$ to $\cal B$.
\end{theorem}

\begin{proof}
This theorem can be proved in a similar way as Theorems \ref{th:gwfs} and \ref{th:gfb}.
\end{proof}

Also, the greatest weak forward bisimulation can be represented as follows.

\begin{corollary}\label{cor:gwbs}
Let ${\cal A}=(A,\delta^A,\sigma^A,\tau^A)$ and ${\cal B}=(B,\delta^B,\sigma^B,\tau^B)$ be automata such that there exists at least one weak forward bisimulation from $\cal A$ to $\cal B$, and let $\mu $ be the greatest weak forward bisimulation from $\cal A$ to $\cal B$.~Then
\begin{equation}\label{eq:gwfsl}
(a,b)\in \mu \ \Leftrightarrow\ \overrightarrow{L}(a)= \overrightarrow{L}(b),
\end{equation}
for all $a\in A$ and $b\in B$.
\end{corollary}

\begin{proof}
We can prove this corollary in a similar way as Corollary \ref{cor:gwfs}.
\end{proof}

Let ${\cal A}=(A,\delta^A,\sigma^A,\tau^A)$ be an arbitrary automaton.~A weak forward~bisimulation from $\cal A$ into itself will be called a {\it weak forward bisimulation on\/} $\cal A$ (analogously we define {\it weak backward bisimu\-la\-tions on\/}~$\cal A$).~The family of all weak forward bisimulations on $\cal A$ is non-empty (it contains at least the~equality relation),~and according to Theorem \ref{th:gwfb}, there is the greatest
weak forward bisimulation on $\cal A$, which is defined as in (\ref{eq:gwfb}), and it is easy to check that it is an equivalence (cf.~\cite{SCI.11}).~Weak forward bisimulations on $\cal A$ which are equivalences will be called {\it weak forward~bisimulation equivalences\/} (analogously we define {\it weak backward bisimu\-la\-tion equivalences\/}).~The set of all weak forward bisimulation equivalences on $\cal A$ we denote by ${\cal E}^{\mathrm{wfb}}({\cal A})$.

Note that condition (\ref{eq:wfs.sigma}) is satisfied whenever $A=B$ and $\varphi $ is a reflexive relation, and hence, it is satisfied whenever $A=B$ and $\varphi $ is an equivalence.~Therefore, an equivalence $E$ on $A$ is a weak forward bisimulation on $\cal A$ if and only if
\begin{equation}\label{eq:wfbe}
E\circ \tau_u^A\subseteq \tau_u^A, \quad\text{for every $u\in X^*$},
\end{equation}
or equivalently,
\begin{equation}\label{eq:wfbe2}
E\circ \tau_u^A= \tau_u^A, \quad\text{for every $u\in X^*$}.
\end{equation}
Analogously, an equivalence $E$ on $A$ is a weak backward bisimulation on $\cal A$ if and only if
\begin{equation}\label{eq:wbbe}
\sigma_u^A\circ E\subseteq \sigma_u^A, \quad\text{for every $u\in X^*$},
\end{equation}
or equivalently,
\begin{equation}\label{eq:wbbe2}
\sigma_u^A\circ E= \sigma_u^A, \quad\text{for every $u\in X^*$}.
\end{equation}

In Theorem \ref{th:lat.fbe} we proved that forward bisimulation equivalences on an automaton form a complete join-subsemilattice of the lattice of equivalences on this automaton.~For weak forward bisimulation equivalences we show even more, that they form a principal ideal of the lattice of equivalences.

\begin{theorem}\label{th:lat.wfbe}
Let ${\cal A}=(A,\delta^A,\sigma^A,\tau^A)$ be an automaton.

The set ${\cal E}^{\mathrm{wfb}}({\cal A})$ of all weak forward bisimulation equivalences on $\cal A$ forms a principal ideal of the lattice ${\cal E}(A)$ of all equivalences on $A$ generated by the relation $E^{\mathrm{wfb}}$ on $A$ defined by
\begin{equation}\label{eq:Ewfb}\
(a,a')\in E^{\mathrm{wfb}} \ \Leftrightarrow\ (\forall u\in X^*)\ (\,a\in \tau_u^A \ \Leftrightarrow\ a'\in\tau_u^A),
\end{equation}
for all $a,a'\in A$.
\end{theorem}

\begin{proof}
It is clear that $E^{\mathrm{wfb}}$ is an equivalence.~For arbitrary $u\in X^*$ and $a\in A$, by $a\in E^{\mathrm{wfb}}\circ\tau_u^A$ it follows that $(a,a')\in E^{\mathrm{wfb}}$ and $a'\in \tau_u^A$, for some $a'\in A$, and by (\ref{eq:Ewfb}) we obtain that $a\in \tau_u^A$. Therefore, $E^{\mathrm{wfb}}\in {\cal E}^{\mathrm{wfb}}({\cal A})$.

Consider an arbitrary $E\in {\cal E}(A)$.~If $E\subseteq E^{\mathrm{wfb}}$, then $E\circ \tau_u^A\subseteq E^{\mathrm{wfb}}\circ \tau_u^A\subseteq \tau_u^A$, so $E\in {\cal E}^{\mathrm{wfb}}({\cal A})$.~Conversely,~let $E \in {\cal E}^{\mathrm{wfb}}({\cal A})$, i.e., $E\circ \tau_u^A\subseteq \tau_u^A$, for each $u\in X^*$.~For arbitrary $(a,a')\in E$ and $u\in X^*$, if $a\in \tau_u^A$, then $a\in E\circ \tau_u^A\subseteq \tau_u^A$, and by symmetry, if $a'\in \tau_u^A$, then $a\in \tau_u^A$.~By this it follows that $(a,a')\in E^{\mathrm{wfb}}$.~Therefore, $E\subseteq E^{\mathrm{wfb}}$ if and only if $E\in {\cal E}^{\mathrm{wfb}}({\cal A})$, and consequently,
${\cal E}^{\mathrm{wfb}}({\cal A})$ is the principal ideal of ${\cal E}(A)$ generated by $E^{\mathrm{wfb}}$.
\end{proof}

\section{Uniform weak forward bisimulations}

In this section we study weak forward bisimulations which are uniform relations.~Within the class~of~uniform relations, weak forward bisimulations can be characterized as follows.

\begin{theorem}\label{th:uwfb0}
Let ${\cal A}=(A,\delta^A,\sigma^A,\tau^A)$ and ${\cal B}=(B,\delta^B,\sigma^B,\tau^B)$ be
automata and let $\varphi \subseteq A\times B$ be a uniform relation.
Then $\varphi $ is a weak forward bisimulation if and only if the following hold:
\begin{align}
\sigma^A\circ\varphi&=\sigma^B\circ\varphi^{-1}\circ\varphi, &
\sigma^A\circ\varphi\circ\varphi^{-1}&=\sigma^B\circ\varphi^{-1}, &&\label{eq:uwfb.sigma} \\
\varphi^{-1}\circ \tau_u^A& = \tau_u^B, \quad \text{for each $u\in X^*$}, &\ \ \tau_u^A&=\varphi\circ \tau_u^B,  \quad\text{for each $u\in X^*$}.& & \label{eq:uwfb0.tau}
\end{align}
\end{theorem}

\begin{proof}
Let $\varphi $ be a weak forward bisimulation. According to (\ref{eq:wfs.sigma})
and (\ref{eq:wfsi.sigma}) we have that
\[
\sigma^A\circ \varphi \subseteq \sigma^B\circ \varphi^{-1}\circ \varphi \subseteq
\sigma^A\circ \varphi \circ \varphi^{-1}\circ \varphi = \sigma^A\circ \varphi ,
\]
and hence, $\sigma^A\circ \varphi =  \sigma^B\circ \varphi^{-1}\circ \varphi$.~In a similar way we prove that $\sigma^B\circ \varphi^{-1}=\sigma^A\circ \varphi\circ \varphi^{-1}$.

Next, by reflexivity of $\varphi^{-1}\circ \varphi $, for each $u\in X^*$ we have that
\[
\tau_u^B\subseteq \varphi^{-1}\circ\varphi \circ \tau_u^B\subseteq \varphi^{-1}\circ \tau_u^A ,
\]
and by this and (\ref{eq:wfs.tau}) we obtain that $\tau_u^B=\varphi^{-1}\circ \tau_u^A$.~Similarly we prove that $\tau_u^A=\varphi\circ \tau_u^B$.

Conversely, let (\ref{eq:uwfb.sigma}) and (\ref{eq:uwfb0.tau}) hold.~It is clear that (\ref{eq:uwfb0.tau}) implies both (\ref{eq:wfs.tau}) and (\ref{eq:wfsi.tau}), and by reflexivity of $\varphi\circ \varphi^{-1}$ and $\varphi^{-1}\circ \varphi$ we obtain that
\[
\sigma^A\subseteq \sigma^A\circ \varphi\circ \varphi^{-1} = \sigma^B\circ \varphi^{-1}, \qquad
\sigma^B\subseteq \sigma^B\circ \varphi^{-1}\circ \varphi = \sigma^A\circ \varphi ,
\]
and hence, (\ref{eq:wfs.sigma}) and (\ref{eq:wfsi.sigma}) hold.~Therefore, $\varphi $ is a weak forward bisimulation.
\end{proof}

Further we prove two very useful lemmas.

\begin{lemma}\label{le:tau.factor}
Let ${\cal A}=(A,\delta^{A},\sigma^{A},\tau^{A})$ be an automaton, $E$ an equivalence on $A$, and ${\cal A}/E=(A/E,\delta^{A/E},\sigma^{A/E},\tau^{A/E})$ the factor automaton of $\cal A$ with respect to $E$.~If $E$ is weak forward bisimulation equivalence, then
\begin{equation}\label{eq:tau.factor}
E_{a}\in \tau^{A/E}_{u} \ \Leftrightarrow\ a\in \tau^{A}_{u},
\end{equation}
for all $u\in X^{*}$ and $a\in A$.
\end{lemma}

\begin{proof}
The claim will be proved by induction on the length of the word $u$.

According to (\ref{eq:tauAE}) and the hypothesis of the lemma, the claim is true if $u$ is the empty word.~Suppose that the claim is true for some word $u$, and consider arbitrary $x\in X$ and $a\in A$.~Then we have that
\[
\begin{aligned}
E_a\in \tau_{xu}^{A/E}= \delta_x^{A/E}\circ \tau_{u}^{A/E}\ &\Leftrightarrow\ (\exists a'\in A)\ (\,(E_a,E_{a'})\in \delta_x^{A/E}\ \land\ E_{a'}\in \tau_{u}^{A/E}\,) \\
&\Leftrightarrow\ (\exists a'\in A)\ (\,(a,{a'})\in E\circ \delta_x^{A}\circ E\ \land\ {a'}\in \tau_{u}^{A}\,) \\
&\Leftrightarrow\ a\in E\circ \delta_x^{A}\circ E\circ \tau_{u}^{A} = E\circ \delta_x^{A}\circ \tau_{u}^{A} = E\circ \tau_{xu}^{A} = \tau_{xu}^{A}.
\end{aligned}
\]
Therefore, the claim is true for all $u\in X^*$ and $a\in A$.
\end{proof}

\begin{lemma}\label{le:nat.wfb}
Let ${\cal A}=(A,\delta^A,\sigma^A,\tau^A)$ be an automaton, let $E$ be an equivalence
on~$A$, let $\varphi_E$ be the natural function from $A$ to $A/E$, and let ${\cal A}/E=(A/E,\delta^{A/E},\sigma^{A/E},\tau^{A/E})$ be the factor automaton
of $\cal A$ with respect to~$E$.

Then $E$ is a weak forward bisimulation equivalence on $\cal A$ if and only if $\varphi_E$ is a weak forward bisimulation between $\cal A$ and ${\cal A}/E$.
\end{lemma}

\begin{proof}
Let $E$ be a weak forward bisimulation equivalence on $\cal A$.~According to Lemma \ref{le:tau.factor}, for arbitrary $u\in X^*$ and $a\in A$ we have that
\[
E_a\in \varphi_E^{-1}\circ \tau_u^A\ \Leftrightarrow\ (\exists a'\in A)\ (\,(E_a,a')\in \varphi_E^{-1} \ \land \ a'\in \tau_u^A \,) \ \Leftrightarrow\ (\exists a'\in A)\ (\,E_a=E_{a'}\ \land \ E_{a'}\in \tau_u^{A/E}\,) \ \Rightarrow\ E_{a}\in \tau_u^{A/E},
\]
and hence, $\varphi_E^{-1}\circ \tau_u^A\subseteq \tau_u^{A/E}$.~Moreover, we have that $\sigma^A\subseteq \sigma^A\circ E$, by reflexivity of $E$, and according to (\ref{eq:sigmaAE}), for each $a\in A$ by $a\in \sigma^A\subseteq \sigma^A\circ E$ it follows $E_a\in \sigma^{A/E}$, and since $(E_a,a)\in \varphi_E^{-1}$, we obtain that $a\in \sigma^{A/E}\circ \varphi_E^{-1}$.~Thus, $\sigma^A\subseteq \sigma^{A/E}\circ \varphi_E^{-1}$.~In the same way we show that $\varphi_E\circ \tau_u^{A/E}\subseteq \tau_u^A$, for each $u\in X^*$, and $\sigma^{A/E}\subseteq \sigma^A\circ \varphi_E$.~Therefore, $\varphi_E$ is a weak forward bisimulation between $\cal A$ and ${\cal A}/E$.

Conversely, let $\varphi_E$ be a weak forward bisimulation between $\cal A$ and ${\cal A}/E$.~According to this assumption and (\ref{eq:tauAE}), for arbitrary $u\in X^*$ and $a\in A$ we have that
\[
a\in E\circ \tau_u^A\ \Leftrightarrow\ E_a\in \tau_u^{A/E}\ \Rightarrow\ (a,E_a)\in \varphi_E\ \land\ E_a\in \tau_u^{A/E}\ \Rightarrow\ a\in \varphi_E\circ \tau_u^{A/E}\subseteq \tau_u^A .
\]
Thus $E\circ \tau_u^A\subseteq \tau_u^A$, and we have proved that $E$ is a weak forward bisimulation equivalence on $\cal A$.
\end{proof}

Let ${\cal A}=(A,\delta^A,\sigma^A,\tau^A)$ and ${\cal B}=(B,\delta^B,\sigma^B,\tau^B)$ be automata,  and let $\phi :A\to B$ be a bijective function. If~$\phi $~satisfies
\begin{align}
a\in \sigma^A\ \Leftrightarrow\ \phi(a)\in \sigma^B, \qquad &\text{for every $a\in A$},&& \label{eq:wfi.sigma}\\
a\in \tau_u^A\ \Leftrightarrow\ \phi(a)\in \tau_u^B, \qquad &\text{for all $u\in X^*$ and $a\in A$},&& \label{eq:wfi.tau}
\end{align}
then it is called a {\it weak forward isomorphism\/} between $\cal A$ and $\cal B$.~Similarly, if
$\phi $~satisfies
\begin{align}
a\in \sigma_u^A\ \Leftrightarrow\ \phi(a)\in \sigma_u^B, \qquad &\text{for all $u\in X^*$ and $a\in A$},&& \label{eq:wbi.sigma}\\
a\in \tau^A\ \Leftrightarrow\ \phi(a)\in \tau^B, \qquad &\text{for every $a\in A$},&& \label{eq:wbi.tau}
\end{align}
then it is called a {\it weak backward isomorphism\/} between $\cal A$ and $\cal B$.~It is easy to check that the inverse function of a weak forward (resp.~backward) isomorphism is also a weak forward (resp.~backward) isomorphism.

Now we state and prove the following analogue of Theorem \ref{th:ufb}.~The main difference is that in this case the factor automata need not be isomorphic, but only weak forward isomorphic.

\begin{theorem}\label{th:uwfb}
Let ${\cal A}=(A,\delta^A,\sigma^A,\tau^A)$ and ${\cal B}=(B,\delta^B,\sigma^B,\tau^B)$ be automata and let $\varphi \subseteq A\times B$ be a uniform relation. Then $\varphi $ is a weak forward bisimulation if and only if the following hold:
\begin{itemize}\parskip=0pt\itemindent6pt
\item[{\rm (i)}] $E_A^\varphi $ is a weak forward bisimulation equivalence on $\cal A$;
\item[{\rm (ii)}] $E_B^\varphi $ is a weak forward bisimulation equivalence on $\cal B$;
\item[{\rm (iii)}]  $\widetilde \varphi $ is a weak forward isomorphism of factor automata ${\cal A}/E_A^\varphi $ and ${\cal B}/E_B^\varphi $.
\end{itemize}
\end{theorem}

\begin{proof}
For the sake of simplicity set $E_A^\varphi = E$ and $E_B^\varphi = F$.~Moreover, let $f\in FD(\varphi )$ be an arbitrary functional description of $\varphi $.

Let $\varphi $ be a weak forward bisimulation.~Then we have that
\[
E\circ \tau_u^A=\varphi\circ \varphi^{-1}\circ \tau_u^A\subseteq \varphi\circ \tau_u^B \subseteq \tau_u^A ,
\]
and since the opposite inclusion follows by reflexivity of $E$, we conclude that $E\circ \tau_u^A=\tau_u^A$.~Hence, $E$ is a weak forward bisimulation equivalence on $\cal A$.~In a similar way we prove that $F$ is a weak forward bisimulation equivalence on $\cal B$.

Next, for an arbitrary $a\in A$ we have that
\[
\begin{aligned}
E_a\in \sigma^{A/E}\ &\Leftrightarrow\ a\in \sigma^A\circ E = \sigma^A\circ \varphi \circ \varphi^{-1} = \sigma^B\circ \varphi^{-1}\ \Leftrightarrow\ (\exists b\in B)\ (\,b\in \sigma^B\ \land\ (a,b)\in \varphi \,) \\
&\Leftrightarrow\  (\exists b\in B)\ (\,b\in \sigma^B\ \land\ (f(a),b)\in F\, )\ \Leftrightarrow \ f(a)\in \sigma^B\circ F \ \Leftrightarrow\ F_{f(a)}\in\sigma^{B/F},
\end{aligned}
\]
and for arbitrary $u\in X^*$ and $a\in A$ we obtain
\[
\begin{aligned}
E_a\in \tau_u^{A/E}\ &\Leftrightarrow\ a\in \tau_u^A =\varphi\circ \tau_u^B \ \Leftrightarrow\ (\exists b\in B)\ (\,(a,b)\in\varphi\ \land\ b\in \tau_u^B \,) \\
&\Leftrightarrow\  (\exists b\in B)\ (\,(f(a),b)\in F\ \land\ b\in \tau_u^B\, )\ \Leftrightarrow \ f(a)\in F\circ \tau_u^B =\tau_u^B\ \Leftrightarrow\ F_{f(a)}\in\tau_u^{B/F}.
\end{aligned}
\]
Therefore, we have proved that $\widetilde \varphi :E_a\mapsto F_{f(a)}$ is a weak forward isomorphism between ${\cal A}/E_A^\varphi $ and ${\cal B}/E_B^\varphi $.

Conversely, let (i), (ii), and (iii) hold. For an arbitrary $a\in A$ we have that
\[
\begin{aligned}
&a\in \sigma^A\circ \varphi\circ \varphi^{-1}=\sigma^A\circ E\ \Leftrightarrow\ E_a\in \sigma^{A/E} \
\Leftrightarrow\ \widetilde\varphi (E_a)\in\sigma^{B/F} \ \Leftrightarrow\
F_{f(a)}\in\sigma^{B/F} \ \Leftrightarrow\ {f(a)}\in\sigma^{B}\circ F \\
&\hspace{15mm}\Leftrightarrow\ (\exists b\in B)\ (\,b\in\sigma^B\ \land\ (b,f(a))\in F \,)\
\Leftrightarrow\ (\exists b\in B)\ (\,b\in\sigma^B\ \land\ (a,b)\in \varphi  \,)
\ \Leftrightarrow\ a\in \sigma^B\circ \varphi^{-1},\end{aligned}
\]
so $\sigma^A\circ \varphi\circ \varphi^{-1}=\sigma^B\circ \varphi^{-1}$, and consequently,
$\sigma^A\circ \varphi = \sigma^A\circ \varphi\circ \varphi^{-1}\circ \varphi =\sigma^B\circ \varphi^{-1}\circ\varphi $.~Moreover, for arbitrary $u\in X^*$ and $a\in A$ we have
\[
\begin{aligned}
a\in \tau_u^A\ &\Leftrightarrow\ E_a\in \tau_u^{A/E} \
\Leftrightarrow\ \widetilde\varphi (E_a)\in\tau_u^{B/F} \ \Leftrightarrow\
F_{f(a)}\in\tau_u^{B/F} \ \Leftrightarrow\ {f(a)}\in\tau_u^B=F\circ\tau_u^{B} \\
&\Leftrightarrow\ (\exists b\in B)\ (\,(f(a),b)\in F \ \land\ b\in\sigma^B\,)\
\Leftrightarrow\ (\exists b\in B)\ (\,(a,b)\in \varphi \ \land\ b\in\sigma^B\,)\
\Leftrightarrow\ a\in \varphi\circ\tau_u^B,
\end{aligned}
\]
so $\tau_u^A=\varphi\circ\tau_u^B$, which also yields $\varphi^{-1}\circ\tau_u^A= \varphi^{-1}\circ\varphi\circ \tau_u^B=F\circ \tau_u^B =\tau_u^B$.~Therefore, according to Theorem \ref{th:uwfb0}, $\varphi $ is a weak forward bisimulation.
\end{proof}

We can also prove the following.

\begin{theorem}\label{th:uwfb.ex}
Let ${\cal A}=(A,\delta^A,\sigma^A,\tau^A)$ and ${\cal B}=(B,\delta^B,\sigma^B,\tau^B)$ be
automata,~and let $E$ and $F$ be weak forward~bisimulation equivalences on $\cal A$ and $\cal B$.

Then there exists a uniform weak forward bisimulation $\varphi\subseteq A\times
B$~such that $E_A^\varphi =E$ and $E_B^\varphi =F$ if and only if there exists a weak forward isomorphism between factor~automata ${\cal A}/E$ and ${\cal B}/F$.
\end{theorem}

\begin{proof}
This theorem can be proved in a similar way as Theorem \ref{th:ufb.ex}, using Theorem \ref{th:uwfb}.
\end{proof}

\begin{theorem}\label{th:G:E.weak}
Let ${\cal A}=(A,\delta^A,\sigma^A,\tau^A)$ be an automaton, let $E$ be a weak forward
bisimulation equivalence on $\cal A$, and let $F$ be an equivalence on $A$ such that $E\subseteq F$.

Then $F$ is a weak forward bisimulation equivalence on $\cal A$ if and only if $F/E$ is a weak forward bisimulation equivalence~on~${\cal A}/E$.
\end{theorem}

\begin{proof}
For arbitrary $u\in X^*$ and $a\in A$ we can easily check that
\[
E_a\in (F/E)\circ \tau_u^{A/E} \ \Leftrightarrow\ a\in F\circ \tau_u^A .
\]
By this and by Lemma \ref{le:tau.factor} we obtain that $(F/E)\circ \tau_u^{A/E}\subseteq \tau_u^{A/E}$ if and only
if $F\circ \tau_u^A\subseteq \tau_u^A$, what is precisely the claim of the theorem.
\end{proof}

\begin{corollary}\label{cor:F:E.g.weak}
Let ${\cal A}=(A,\delta^A,\sigma^A,\tau^A)$ be an automaton, and let $E$ and $F$ be
weak forward bisimulation equivalences on $\cal A$ such that $E\subseteq F$.

Then $F$ is the greatest weak forward bisimulation equivalence on $\cal A$ if and only if $F/E$ is the greatest weak forward bisimulation equivalence on ${\cal A}/E$.
\end{corollary}

\begin{proof}
This is an immediate consequence of the previous theorem and Theorem \ref{th:F:E-isom}.
\end{proof}

Let ${\cal A}=(A,\delta^A,\sigma^A,\tau^A )$ be an automaton.~Let~us set $A_N =\{\sigma_u^A\mid u\in X^*\}$,
and let us define $\delta^{A_N} :A_N\times X\to A_N $ and $\tau^{A_N}\subseteq A_N$ by
\begin{align}
&\delta^{A_N} (\sigma_u^A,x)=\sigma_{ux}^A, \label{eq:delta.N}\\
&\sigma_u^A\in \tau^{A_N} \ \Leftrightarrow\ \sigma_u^A\circ \tau^A =1 \ \Leftrightarrow\ \sigma_u^A\cap \tau^A \ne\emptyset,\label{eq:tau.N}
\end{align}
for all $u\in X^*$ and $x\in X$.~Then ${\cal A}_N=(A_N,\delta^{A_N},\sigma_\varepsilon^A,\tau^{A_N})$ is a deterministic automaton which is language~equivalent to $\cal A$, i.e., $L({\cal A}_N)=L({\cal A})$, and it is called the {\it Nerode automaton\/} of $\cal A$ (cf.~\cite{CDIV.10,ICB.08,ICBP.10,JIC.11}).~Note that the Nerode automaton of $\cal A$ is the deterministic automaton obtained from $\cal A$ by means of the determinization method known as the {\it accessible subset construction\/}.

Moreover, let $\bar{A}_N =\{\tau_u^A\mid u\in X^*\}$,
and let us define $\delta^{\bar{A}_N} :\bar{A}_N\times X\to \bar{A}_N $ and $\tau^{\bar{A}_N}\subseteq \bar{A}_N$ by
\begin{align}
&\delta^{\bar{A}_N} (\tau_u^A,x)=\tau_{xu}^A, \label{eq:delta.N.r}\\
&\tau_u^A\in \tau^{\bar{A}_N} \ \Leftrightarrow\ \sigma^A\circ \tau_u^A=1 \ \Leftrightarrow\ \sigma^A\cap \tau_u^A \ne\emptyset,\label{eq:tau.N.r}
\end{align}
for all $u\in X^*$ and $x\in X$.~Then $\bar{\cal A}_N =(\bar{A}_N,\delta^{\bar{A}_N},\tau_\varepsilon^A,\tau^{\bar{A}_N})$ is a deterministic automaton which is isomorphic to the Nerode automaton of the reverse automaton $\bar{\cal A}$ of $\cal A$, and it is called the {\it reverse Nerode automaton\/} of $\cal A$.

The following theorem gives a characterization of uniform weak forward bisimulations in terms of the reverse Nerode automata.~Let us note that an analogous theorem, given in terms of the Nerode automata, characterizes uniform weak backward bisimulations.

\begin{theorem}\label{th:uwfb.N}
Let ${\cal A}=(A,\delta^A,\sigma^A,\tau^A)$ and ${\cal B}=(B,\delta^B,\sigma^B,\tau^B)$ be automata and $\varphi \subseteq A\times B$ a uniform relation.

Then $\varphi $ is a weak forward bisimulation from $\cal A$ to $\cal B$ if and only if it satisfies (\ref{eq:wfs.sigma}) and (\ref{eq:wfsi.sigma}), and functions
\begin{equation}\label{eq:isom.rN}
\tau_u^A\mapsto \varphi^{-1}\circ \tau_u^A, \qquad \tau_u^B\mapsto \varphi\circ \tau_u^B,
\end{equation}
for each $u\in X^*$, are mutually inverse isomorphisms between reverse Nerode automata $\bar{\cal A}_N$ and $\bar{\cal B}_N$.
\end{theorem}

\begin{proof}
Consider functions $\Phi :{\bar A}_N\to {\cal P}(B)$ and $\Psi :{\bar B}_N\to {\cal P}(A)$ which are given by
$\Phi(\tau_u^A)=\varphi^{-1}\circ \tau_u^A$ and $\Psi(\tau_u^B)=\varphi\circ \tau_u^B$, for each $u\in X^*$.~

Let $\varphi $ be a weak forward bisimulation from $\cal A$ to $\cal B$. By definition, it satisfies (\ref{eq:wfs.sigma}) and (\ref{eq:wfsi.sigma}).~According~to Theorem \ref{th:uwfb0}, for every $u\in X^*$ we have that $\Phi (\tau_u^A) = \tau_u^B\in {\bar B}_N$ and $\Psi (\tau_u^B)=\tau_u^A\in {\bar A}_N$, which means that $\Phi $ maps
${\bar A}_N$ into ${\bar B}_N$, and $\Psi $ maps ${\bar B}_N$ into ${\bar A}_N$.~According to the same theorem, for every $u\in X^*$ we have that $\Psi (\Phi (\tau_u^A)) = \varphi\circ\varphi^{-1}\circ \tau_u^A= \tau_u^B$ and $\Phi (\Psi (\tau_u^B)) = \varphi^{-1}\circ\varphi \circ \tau_u^B= \tau_u^B$, and hence, $\Phi $ and $\Psi $ are mutually inverse bijections from ${\bar A}_N$ to ${\bar B}_N$, and vice versa.

Clearly, $\Phi (\tau^A)=\tau^B$ and $\Psi(\tau^B)=\tau^A$.~Next, for arbitrary $x\in X$ and $u\in X^*$ we have that
\[
\Phi(\delta^{{\bar A}_N}(\tau_u^A,x))=\Phi (\tau_{xu}^A) = \tau_{xu}^B = \delta^{{\bar B}_N}(\tau_u^B,x) = \delta^{{\bar B}_N}(\Phi(\tau_u^A),x).
\]
By Theorem \ref{th:uwfb0}, for any $u\in X^*$ we have that $\sigma^A\circ \tau_u^A=\sigma^A\circ \varphi \circ \varphi^{-1}\circ \tau_u^A = \sigma^B\circ \varphi^{-1}\circ \tau_u^A = \sigma^B\circ \tau_u^B$,~so
\[
\tau_u^A\in \tau^{{\bar A}_N} \ \Leftrightarrow\ \sigma^A\circ \tau_u^A = 1 \ \Leftrightarrow\  \sigma^B\circ \tau_u^B = 1  \ \Leftrightarrow\ \tau_u^B\in \tau^{{\bar B}_N} \ \Leftrightarrow\ \Phi(\tau_u^A)\in \tau^{{\bar B}_N}.
\]
Hence, we have proved that $\Phi $ is an isomorphism from $\bar{\cal A}_N$ to $\bar{\cal B}_N$.~In a similar way we prove that $\Psi $ is an isomorphism from $\bar{\cal B}_N$ to $\bar{\cal A}_N$.

Conversely, let (\ref{eq:wfs.sigma}) and (\ref{eq:wfsi.sigma}) hold, and let $\Phi $ and $\Psi $ be mutually inverse isomorphisms from $\bar{\cal A}_N$ to $\bar{\cal B}_N$ and from $\bar{\cal B}_N$ to $\bar{\cal A}_N$, respectively.~Since $\tau^A$ and $\tau^B$ are the unique initial states of $\bar{\cal A}_N$ and $\bar{\cal B}_N$, we have that $\Phi (\tau^A)=\tau^B$, and hence, $\varphi^{-1}\circ \tau^A=\tau^B$ and $\varphi\circ \tau^B=\tau^A$.~Suppose that $\Phi(\tau_u^A)=\tau_u^B$, for some $u\in X^*$, and~consider an arbitrary $x\in X$.~Then
\[
\Phi(\tau_{xu}^A)=\Phi(\delta^{\bar{A}_N}(\tau_u^A,x)) = \delta^{\bar{B}_N}(\Phi(\tau_u^A),x) =
\delta^{\bar{B}_N}(\tau_u^B,x) = \tau_{xu}^B.
\]
Now, by induction on the length of $u$ we obtain that $\Phi(\tau_u^A)=\tau_u^B$, for every $u\in X^*$, and also, $\Psi(\tau_u^B)=\tau_u^A$, which means that (\ref{eq:uwfb0.tau}) holds.~Therefore, by Theorem \ref{th:uwfb0} we obtain that $\varphi $ is a weak forward bisimulation.
\end{proof}

Note that a similar theorem can be proved for weak backward bisimulations, i.e., a uniform relation
$\varphi $ is a weak backward bisimulation from $\cal A$ to $\cal B$ if and only if it satisfies (\ref{eq:wbs.tau}) and (\ref{eq:wbsi.tau}), and functions
\begin{equation}\label{eq:isom.rN}
\sigma_u^A\mapsto \sigma_u^A\circ \varphi, \qquad \sigma_u^B\mapsto \sigma_u^B\circ\varphi^{-1} ,
\end{equation}
for each $u\in X^*$, are mutually inverse isomorphisms between Nerode automata ${\cal A}_N$ and ${\cal B}_N$.

\section{Weak forward bisimulation equivalent automata}

Let ${\cal A}=(A,\delta^A,\sigma^A,\tau^A)$ and ${\cal B}=(B,\delta^B,\sigma^B,\tau^B)$ be
automata.~If there exists~a complete and surjective weak forward bisimula\-tion from
$\cal A$ to $\cal B$, then we say that $\cal A$ and $\cal B$ are {\it weak forward bisimulation~equi\-valent\/},
or briefly {\it WFB-equivalent\/}, and we write ${\cal A}\sim_{WFB}{\cal B}$.~Notice that completeness
and surjectivity of this forward bisimulation mean that every state of $\cal A$ is equivalent to some state
of $\cal B$, and vice versa.~For~arbitrary~auto\-mata $\cal A$, $\cal B$
and $\cal C$ we have that
\begin{equation}\label{eq:wfbe.eq}
{\cal A}\sim_{WFB}{\cal A};\ \ \ \ {\cal A}\sim_{WFB}{\cal B} \implies {\cal B}\sim_{WFB}{\cal A};\ \ \ \
\bigl({\cal A}\sim_{WFB}{\cal B} \land {\cal B}\sim_{WFB}{\cal C}\bigr) \implies {\cal A}\sim_{WFB}{\cal C}.
\end{equation}
Similarly, we say that $\cal A$ and $\cal B$ are {\it weak backward bisimulation equivalent\/}, briefly {\it WBB-equiva\-lent\/}, in notation
${\cal A}\sim_{WBB}{\cal B}$, if there exists a complete and surjective weak backward bisimulation from $\cal A$ to $\cal B$.

The following lemma will be useful in our further work.

\begin{lemma}\label{le:wf.isom}
Let ${\cal A}=(A,\delta^A,\sigma^A,\tau^A)$ and ${\cal B}=(B,\delta^B,\sigma^B,\tau^B)$ be automata, let $\phi $ be a weak forward isomorphism between $\cal A$ and $\cal B$, and let $E$ and $F$ be the greatest weak forward bisimulation equivalences on $\cal A$ and $\cal B$.

Then for arbitrary $a_1,a_2\in A$ the following is true:
\begin{equation}\label{eq:EphiF}
(a_1,a_2)\in E  \ \ \Leftrightarrow\ \ (\phi(a_1),\phi(a_2))\in F .
\end{equation}
\end{lemma}

\begin{proof}
Let us define a relation $F'$ on $B$ by
\begin{equation}\label{eq:E'phiF}
(b_1,b_2)\in F'  \ \ \Leftrightarrow\ \ (\phi^{-1}(b_1),\phi^{-1}(b_2))\in E ,
\end{equation}
for arbitrary $b_1,b_2\in B$.~It is clear that $F'$ is an equivalence on $B$.

Consider an arbitrary $u\in X^*$.~If $b_1\in F'\circ \tau_u^B$, then there is $b_2\in B$ such that $(b_1,b_2)\in F'$ and $b_2\in \tau_u^B $, and by (\ref{eq:E'phiF}) and (\ref{eq:wfi.tau}) we obtain that $(\phi^{-1}(b_1),\phi^{-1}(b_2))\in E$ and $\phi^{-1}(b_2)\in \tau_u^A$.~This means that $\phi^{-1}(b_1)\in E\circ \tau_u^A\subseteq \tau_u^A$, and again by  (\ref{eq:wfi.tau}) we obtain that $b_1=\phi (\phi^{-1}(b_1))\in \tau_u^B$.~Therefore, $F'\circ \tau_u^B\subseteq \tau_u^B$, for each $u\in X^*$, so $F'$ is a weak forward bisimulation equivalence on $\cal A$, whence $F'\subseteq F$.~Now, for arbitrary $a_1,a_2\in A$ we have that $(a_1,a_2)\in E$ implies $(\phi(a_1),\phi(a_2))\in F'\subseteq F$, so we have proved the direct implication in (\ref{eq:EphiF}).~Analogously we prove the reverse implication.
\end{proof}

Now we state and prove the main result of this section.

\begin{theorem}\label{th:UWFBeq}
Let ${\cal A}=(A,\delta^A,\sigma^A,\tau^A)$ and
${\cal B}=(B,\delta^B,\sigma^B,\tau^B)$ be automata, and let $E$ and $F$ be
the greatest weak forward bisimulation equivalences on $\cal A$ and $\cal B$.

Then $\cal A$ and $\cal B$ are WFB-equivalent if and only if there exists a weak forward isomorphism between factor automata ${\cal A}/E$ and ${\cal B}/F$.
\end{theorem}

\begin{proof}
Let $\cal A$ and $\cal B$ be WFB-equivalent automata.~As in the proof of Theorem \ref{th:UFBeq} we show that the greatest weak forward bisimulation $\varphi $ between $\cal A$ and $\cal B$ is a uniform relation.

By Theorem \ref{th:uwfb}, $E_A^\varphi $ and $E_B^\varphi $ are weak forward bisimulation equivalences on $\cal A$ and $\cal B$, and $\widetilde \varphi $ is a weak forward isomorphism of factor automata ${\cal A}/E_A^\varphi $ and ${\cal B}/E_B^\varphi $.~Let $P$ and $Q$ be respectively the greatest weak forward bisimulation equivalences on ${\cal A}/E_A^\varphi $ and ${\cal B}/E_B^\varphi $.~Let $\xi :({\cal A}/E_A^\varphi )/P\to ({\cal B}/E_B^\varphi )/Q$ be a function defined by $\xi (P_\alpha)=Q_{\widetilde\varphi(\alpha)}$, for each $\alpha \in A/E_A^\varphi $.~It is easy to verify that $\xi $ is a well-defined bijective function, and by (\ref{eq:tau.factor}), (\ref{eq:EphiF}) and the fact that $\widetilde\varphi $ is a weak forward isomorphism we obtain that $\xi $ is a weak forward isomorphism.

By Corollary \ref{cor:F:E.g.weak} it follows that $P=E/E_A^\varphi $ and $Q=F/E_B^\varphi $, and according to Theorem \ref{th:F:E}, ${\cal A}/E$ is~isomorphic to $({\cal A}/E_A^\varphi)/P$ and ${\cal B}/F$ is isomorphic to $({\cal B}/E_B^\varphi)/Q$.~As we have already proved that $\xi $  is a weak forward isomorphism between $({\cal A}/E_A^\varphi)/P$ and $({\cal B}/E_B^\varphi)/Q$, we conclude that there is a weak forward isomorphism between ${\cal A}/E$ and ${\cal B}/F$.

The converse follows immediately by Theorem \ref{th:uwfb.ex}.
\end{proof}

\begin{corollary}\label{cor:class.min.weak}
Let $\cal A$ be an automaton, let $E$ be the greatest weak forward bisimulation
equivalence on $\cal A$, and let $\Bbb{WFB}(A)$ be the class of all automata which are WFB-equivalent~to~$\cal A$.

Then ${\cal A}/E$ is a minimal automaton in $\Bbb{WFB}(A)$.~Moreover, if $\cal B$ is any minimal automaton in $\Bbb{WFB}(A)$, then there exists a weak forward isomorphism between ${\cal A}/E$ and $\cal B$.
\end{corollary}

\begin{proof}
Let $\cal B$ be an arbitrary minimal automaton in $\Bbb{WFB}(A)$, and let $F$ be the greatest weak forward bisimulation equivalence on $\cal B$.~According to Theorem \ref{th:UWFBeq}, there exists a weak forward isomorphism between ${\cal A}/E$ and ${\cal B}/F$, and by Lemma \ref{le:nat.wfb} and (\ref{eq:wfbe.eq}) it follows that ${\cal B}/F\in \Bbb{WFB}(A)$.~Now, by minimality of $\cal B$ we obtain that $F$ is the equality relation on $B$, what means that ${\cal B}/F\cong {\cal B}$.~Therefore, there is a weak forward isomorphism between ${\cal A}/F$ and $\cal B$, and consequently, ${\cal A}/F$ is also a minimal automaton in $\Bbb{WFB}(A)$.
\end{proof}

The next example shows that there are automata which are WFB-equivalent, but they are not FB-equi\-valent, and also, that there are automata which are language-equivalent, but they are not WFB-equivalent.

\begin{example}\rm
Let ${\cal A}=(A,\delta^A,\sigma^A,\tau^A)$ and ${\cal B}=(B,\delta^B,\sigma^B,\tau^B)$ be automata with $|A|=4$, $|B|=2$ and $X=\{x\}$, whose transition relations and sets of initial and terminal states are given by the following Boolean matrices and vectors:
\[
\sigma^A=\begin{bmatrix}
0 & 1 & 0 & 0
\end{bmatrix},\ \ \
\delta_x^A=\begin{bmatrix}
1 & 0 & 0 & 0 \\
0 & 0 & 0 & 1 \\
0 & 0 & 0 & 0 \\
0 & 0 & 0 & 0
\end{bmatrix},\ \ \
\tau^A=\begin{bmatrix}
0 \\
0 \\
1 \\
0
\end{bmatrix}, \qquad\quad
\sigma^B=\begin{bmatrix}
1 & 0
\end{bmatrix},\ \ \
\delta_x^B=\begin{bmatrix}
1 & 0 \\
1 & 0
\end{bmatrix},\ \ \
\tau^B=\begin{bmatrix}
0 \\
1
\end{bmatrix}.
\]
Computing the relation $\mu \subseteq A\times B$ using formula (\ref{eq:gwfb}) we obtain that
\[
\mu =\begin{bmatrix}
1 & 0 \\
1 & 0 \\
0 & 1 \\
1 & 0
\end{bmatrix},
\]
and we can easily check that $\mu $ satisfies both (\ref{eq:wfs.sigma}) and (\ref{eq:wfsi.sigma}), and according to Theorem \ref{th:gwfb}, $\mu $ is the greatest weak forward bisimulation between automata $\cal A$ and $\cal B$.

On the other hand, using the procedure from Theorem \ref{th:gfb.alg} we get the relation
\[
\varphi =\begin{bmatrix}
1 & 0 \\
0 & 0 \\
0 & 0 \\
0 & 0
\end{bmatrix}
\]
which does not satisfy (\ref{eq:fs.nda.s}) and (\ref{eq:fb.nda.si}), and according to Theorem \ref{th:gfb.alg}, there is no any forward bisimulation~bet\-ween $\cal A$ and $\cal B$.~Since $\mu $ is complete and surjective (i.e., it is a uniform relation), we have that $\cal A$ and $\cal B$ are WFB-equivalent, but they are not FB-equivalent.

If we change $\sigma^A$ and $\sigma^B$ to
\[
\sigma^A =\begin{bmatrix}
0 & 0 & 1 & 0 \end{bmatrix}, \qquad
\sigma^B =\begin{bmatrix}
1 & 1 \end{bmatrix},
\]
then we obtain that $\mu $ does not satisfy (\ref{eq:wfsi.sigma}), and in this case there is no any weak forward bisimulation~bet\-ween $\cal A$ and $\cal B$, i.e., $\cal A$ and $\cal B$ are not WFB-equivalent.~However, $\cal A$ and $\cal B$ are still language-equivalent, i.e., we have that $L({\cal A})=L({\cal B})\,(=\{\varepsilon\})$.
\end{example}

\section{Concluding remarks}

In this article we have formed a conjunction of bisimulations and uniform relations as a very~power\-ful tool in the study of equivalence between nondeterministic automata.~In this symbiosis, uniform relations serve as equivalences which relate elements of two possibly different sets, while bisimulations provide compatibility with the transitions, initial and terminal states of automata.We have defined six types of bisimulations, but due to the duality we
have discussed three of them: forward, backward-forward,~and weak forward bisimulations. For each od these three types of bisimulations we have provided a procedure which decides whether there is a bisimulation of this type between two automata, and when it exists, the same procedure computes the greatest one.~We have proved that a uniform relation between automata $\cal A$ and $\cal B$ is a forward bisimulation if and only if its kernel and co-kernel are forward bisimulation equivalences on $\cal A$ and $\cal B$ and there is a special isomorphism between factor automata with~respect to these equivalences. As a consequence we get that automata $\cal A$ and $\cal B$ are FB-equivalent, i.e., there is a uniform forward~bisimu\-lation between them, if and only if there is an isomorphism between the factor automata of $\cal A$ and $\cal B$~with respect to their greatest forward bisimulation equivalences.~This result reduces the problem of testing~FB-equi\-valence to the problem of testing isomorphism of automata, which is equivalent to the well-known~graph isomorphism problem.~We have shown that some similar results are also valid for backward-forward~bisim\-ulations, but there are many significant differences.~Analogous results have been also obtained for weak~for\-ward bisimulations, for which we have shown that they are more general than forward bisimulations, and consequently, the WFB-equivalence of automata is closer to the language-equivalence than the FB-equi\-valence.

Similar methodology was used in \cite{CIDB.11} in the study of bisimulations between fuzzy automata.~In~further research, the methodology developed for nondeterministic and fuzzy automata will be applied to weighted automata over suitable types of semirings, as well as in discussing certain issues of social network analysis.

\end{document}